\documentclass[twocolumn,10pt]{article} 
\usepackage[margin=2cm]{geometry} 
\usepackage[subtle, mathspacing=normal]{savetrees}
\usepackage{authblk}

\usepackage{mathtools}

\DeclarePairedDelimiter{\paren}{(}{)}

\usepackage{xspace}
\usepackage{enumitem}
\usepackage{subcaption}

\usepackage{float}

\usepackage{tikz,pgfplots}
\usepackage{etoolbox}
\pgfplotsset{
  every  tick/.style={red,}, minor x tick num=1,
  cycle list={teal,every mark/.append style={fill=teal!80!black},mark=*\\%
orange,every mark/.append style={fill=orange!80!black},mark=square*\\%
cyan!60!black,every mark/.append style={fill=cyan!80!black},mark=otimes*\\%
red!70!white,mark=star\\%
lime!80!black,every mark/.append style={fill=lime},mark=diamond*\\%
red,densely dashed,every mark/.append style={solid,fill=red!80!black},mark=*\\%
yellow!60!black,densely dashed,
every mark/.append style={solid,fill=yellow!80!black},mark=square*\\%
black,every mark/.append style={solid,fill=gray},mark=otimes*\\%
blue,densely dashed,mark=star,every mark/.append style=solid\\%
red,densely dashed,every mark/.append style={solid,fill=red!80!black},mark=diamond*\\%
}
}

\def \CILKtable {
\begin{tikzpicture}[scale = .8]
\begin{axis}[
width = 5 in,
height = 4in,
title={Speedup Versus Number of Threads},
xtick pos=left,
ytick pos=left,
legend style={draw=none},
axis line style = { draw = none },
legend pos= north west,
xtick = data,
xlabel={Number of Threads},
ylabel={Speedup Over Serial Partition},
ymax = 12,
legend columns = 2,
scatter/classes=%
{a={mark=o,draw=blue}}]
\addplot coordinates {( 1, 0.499048) ( 2, 0.995494) ( 3, 1.42948) ( 4, 1.86113) ( 5, 2.26614) ( 6, 2.65797) ( 7, 2.97174) ( 8, 3.26921) ( 9, 3.43004) ( 10, 3.62281) ( 11, 3.78679) ( 12, 3.89481) ( 13, 4.0002) ( 14, 4.0986) ( 15, 4.15578) ( 16, 4.2191) ( 17, 4.22999) ( 18, 4.27599) };
\addplot coordinates {( 1, 0.35092) ( 2, 0.636608) ( 3, 0.894292) ( 4, 1.11755) ( 5, 1.32674) ( 6, 1.5096) ( 7, 1.65121) ( 8, 1.77494) ( 9, 1.88432) ( 10, 1.97837) ( 11, 2.0581) ( 12, 2.11591) ( 13, 2.15699) ( 14, 2.18845) ( 15, 2.21159) ( 16, 2.21657) ( 17, 2.20736) ( 18, 2.21557) };
\addplot coordinates {( 1, 0.223968) ( 2, 0.398625) ( 3, 0.551717) ( 4, 0.691607) ( 5, 0.812203) ( 6, 0.912936) ( 7, 0.990429) ( 8, 1.05055) ( 9, 1.09777) ( 10, 1.13562) ( 11, 1.15517) ( 12, 1.18848) ( 13, 1.20383) ( 14, 1.21364) ( 15, 1.22103) ( 16, 1.22575) ( 17, 1.22285) ( 18, 1.23444) };
\addplot coordinates {( 1, 0.619888) ( 2, 1.24099) ( 3, 1.80558) ( 4, 2.38286) ( 5, 2.95654) ( 6, 3.52917) ( 7, 4.09348) ( 8, 4.63922) ( 9, 5.19034) ( 10, 5.72539) ( 11, 6.27145) ( 12, 6.80311) ( 13, 7.34442) ( 14, 7.85497) ( 15, 8.24706) ( 16, 8.57062) ( 17, 8.77723) ( 18, 9.5674) };
\addplot coordinates {( 1, 0.767408) ( 2, 1.5359) ( 3, 2.21084) ( 4, 2.89005) ( 5, 3.56954) ( 6, 4.2382) ( 7, 4.9214) ( 8, 5.58393) ( 9, 6.23961) ( 10, 6.90587) ( 11, 7.54451) ( 12, 8.18186) ( 13, 8.83243) ( 14, 9.4207) ( 15, 9.88984) ( 16, 10.3425) ( 17, 10.4524) ( 18, 11.1205) };
\legend{Low-Space, Med-Space, High-Space, Smoothed-Striding, Strided}
\end{axis}
\end{tikzpicture}
}

\def \partitionbandwidthboundtable {
\begin{tikzpicture}[scale = .8]
\begin{axis}[
width = 5 in,
height = 4in,
title={Speedup Versus Number of Threads},
xtick pos=left,
ytick pos=left,
legend style={draw=none},
axis line style = { draw = none },
legend pos= north west,
xtick = data,
xlabel={Number of Threads},
ylabel={Speedup Over Serial Partition},
ymax = 20,
legend columns = 2,
scatter/classes=%
{a={mark=o,draw=blue}}]
\addplot coordinates {( 1, 0.499378) ( 2, 0.995949) ( 3, 1.4328) ( 4, 1.86816) ( 5, 2.26982) ( 6, 2.66432) ( 7, 2.98407) ( 8, 3.26602) ( 9, 3.43066) ( 10, 3.61544) ( 11, 3.79325) ( 12, 3.91949) ( 13, 4.00244) ( 14, 4.08389) ( 15, 4.15375) ( 16, 4.22151) ( 17, 4.25438) ( 18, 4.25622) };
\addplot coordinates {(1, 0.894939)(2, 1.72634)(3, 2.26699)(4, 2.66131)(5, 3.15777)(6, 3.4024)(7, 3.50312)(8, 3.49686)(9, 3.62955)(10, 3.74294)(11, 3.90093)(12, 3.9593)(13, 4.0131)(14, 4.04714)(15, 4.13799)(16, 4.13782)(17, 4.10515)(18, 4.15355)};
\addplot coordinates {( 1, 0.62101) ( 2, 1.24639) ( 3, 1.81072) ( 4, 2.37852) ( 5, 2.95182) ( 6, 3.53233) ( 7, 4.09153) ( 8, 4.63759) ( 9, 5.18124) ( 10, 5.73746) ( 11, 6.27769) ( 12, 6.82206) ( 13, 7.35802) ( 14, 7.81717) ( 15, 8.17117) ( 16, 8.81183) ( 17, 8.81973) ( 18, 9.18636) };
\addplot coordinates {(1, 3.52223)(2, 6.68603)(3, 8.70095)(4, 10.276)(5, 12.1258)(6, 12.8292)(7, 13.4991)(8, 13.4622)(9, 14.1134)(10, 14.1066)(11, 14.8215)(12, 15.3338)(13, 15.3938)(14, 15.3942)(15, 15.8412)(16, 15.5877)(17, 15.9778)(18, 15.9519)};
\legend{Low-Space, Low-Space Bandwidth Constraint, Smoothed Striding, Smoothed Striding Bandwidth Constraint}
\end{axis}
\end{tikzpicture}
}

\def \CILKsorttable {
\begin{tikzpicture}[scale = .8]
\begin{axis}[
width = 5 in,
height = 4in,
title={Speedup Versus Number of Threads},
xtick pos=left,
ytick pos=left,
legend style={draw=none},
axis line style = { draw = none },
legend pos= north west,
xtick = data,
xlabel={Number of Threads},
ylabel={Speedup Over Serial Partition},
ymax = 15,
legend columns = 2,
scatter/classes=%
{a={mark=o,draw=blue}}]
\addplot coordinates {( 1, 0.87864) ( 2, 1.68661) ( 3, 2.39404) ( 4, 3.03172) ( 5, 3.71549) ( 6, 4.45118) ( 7, 5.07595) ( 8, 5.76618) ( 9, 6.26716) ( 10, 6.79929) ( 11, 7.34223) ( 12, 7.74736) ( 13, 8.17148) ( 14, 8.70784) ( 15, 8.91736) ( 16, 9.26139) ( 17, 9.47532) ( 18, 9.68665)}; 
\addplot coordinates {( 1, 0.921182) ( 2, 1.80562) ( 3, 2.57745) ( 4, 3.37831) ( 5, 4.15358) ( 6, 4.91613) ( 7, 5.7093) ( 8, 6.41041) ( 9, 7.14098) ( 10, 7.907) ( 11, 8.6573) ( 12, 9.46479) ( 13, 10.1098) ( 14, 10.7975) ( 15, 11.4046) ( 16, 12.1549) ( 17, 12.7427) ( 18, 13.346)}; 
\addplot coordinates {( 1, 0.949951) ( 2, 1.87888) ( 3, 2.70758) ( 4, 3.52248) ( 5, 4.40049) ( 6, 5.19675) ( 7, 6.02708) ( 8, 6.79362) ( 9, 7.69896) ( 10, 8.36354) ( 11, 9.17828) ( 12, 9.92823) ( 13, 10.7114) ( 14, 11.4187) ( 15, 12.0935) ( 16, 12.8171) ( 17, 13.5719) ( 18, 14.2507) };
\legend{Low-Space, Smoothed Striding, Strided}
\end{axis}
\end{tikzpicture}
}

\def \serialtable {
\begin{tikzpicture}[scale = .8]
\begin{axis}[
width = 5 in,
height = 4in,
title={Slowdown Versus Input Size in Serial},
xtick pos=left,
ytick pos=left,
ymax = 6,
ymin = 0,
legend style={draw=none},
axis line style = { draw = none },
legend pos= north west,
xtick = data,
xlabel={Log Input Size},
ylabel={Slowdown Over Serial Partition},
legend columns = 2,
scatter/classes=%
{a={mark=o,draw=blue}}]
\addplot coordinates {( 23, 1.82895) ( 24, 1.81639) ( 25, 1.83007) ( 26, 1.83905) ( 27, 1.84722) ( 28, 1.89343) ( 29, 1.90068) ( 30, 1.91837) };
\addplot coordinates {( 23, 2.75658) ( 24, 2.71803) ( 25, 2.68137) ( 26, 2.67157) ( 27, 2.67443) ( 28, 2.70131) ( 29, 2.70551) ( 30, 2.71713) };
\addplot coordinates {( 23, 4.31579) ( 24, 4.26885) ( 25, 4.20098) ( 26, 4.18464) ( 27, 4.18546) ( 28, 4.22519) ( 29, 4.22257) ( 30, 4.23265) };
\addplot coordinates {( 23, 1.55263) ( 24, 1.54426) ( 25, 1.53431) ( 26, 1.53431) ( 27, 1.53554) ( 28, 1.53491) ( 29, 1.55852) ( 30, 1.53317) };
\addplot coordinates {( 23, 1.53289) ( 24, 1.53443) ( 25, 1.53105) ( 26, 1.52778) ( 27, 1.53064) ( 28, 1.52613) ( 29, 1.5246) ( 30, 1.52919) };
\legend{Low-Space, Med-Space, High-Space, Smoothed Striding, Strided}
\end{axis}
\end{tikzpicture}
}

\newcommand{\dec}{\operatorname{dec}}
\newcommand{\poly}{\operatorname{poly}}
\newcommand{\polylog}{\operatorname{polylog}}

\newcommand{\defn}[1]{{\textit{\textbf{\boldmath #1}}}\xspace}
\renewcommand{\paragraph}[1]{\vspace{0.09in}\noindent{\bf \boldmath #1.}} 
\usepackage{amsmath}
\usepackage{dsfont}
\def\E{\operatorname{\mathbb{E}}}
\usepackage{amssymb}
\usepackage{amsthm}
\usepackage{todonotes}
\usepackage{algorithm, algpseudocode}

\newtheorem{thm}{Theorem}[section]

\theoremstyle{remark}

\newtheorem{theorem}{Theorem}[section]

\newtheorem{lemma}[thm]{Lemma}
\newtheorem{proposition}[thm]{Proposition}

\newtheorem{corollary}[thm]{Corollary}

\theoremstyle{remark}

\usepackage{hyperref}

\usepackage{fancyhdr}
\title{In-Place Parallel-Partition Algorithms \\ using Exclusive-Read-and-Write Memory}
\date{\vspace{-5ex}}

\author[1]{\small William Kuszmaul\thanks{Supported by a Hertz fellowship and a NSF GRFP fellowship}}
\author[1]{\small Alek Westover\thanks{Supported by MIT}}

\affil[1]{\footnotesize MIT}
\affil[ ]{\textit{kuszmaul@mit.edu, alek.westover@gmail.com}}

\begin{document}
  \maketitle

\begin{abstract}
We present an in-place algorithm for the parallel partition
problem that has linear work and polylogarithmic span. The
algorithm uses only exclusive read/write shared variables, and
can be implemented using parallel-for-loops without any
additional concurrency considerations (i.e., the algorithm is
EREW). A key feature of the algorithm is that it exhibits
provably optimal cache behavior, up to small-order factors.

We also present a second in-place EREW algorithm that has linear
work and span $O(\log n \cdot \log \log n)$, which is within an
$O(\log\log n)$ factor of the optimal span. By using this
low-span algorithm as a subroutine within the cache-friendly
algorithm, we are able to obtain a single EREW algorithm that
combines their theoretical guarantees: the algorithm achieves
span $O(\log n \cdot \log \log n)$ and optimal cache behavior. As
an immediate consequence, we also get an in-place EREW quicksort
algorithm with work $O(n \log n)$, span $O(\log^2 n \cdot \log
\log n)$.

Whereas the standard EREW algorithm for parallel partitioning is
memory-bandwidth bound on large numbers of cores, our
cache-friendly algorithm is able to achieve near-ideal scaling in
practice by avoiding the memory-bandwidth bottleneck. The
algorithm's performance is comparable to that of the Blocked
Strided Algorithm of Francis, Pannan, Frias, and Petit, which is
the previous state-of-the art for parallel EREW sorting
algorithms, but which lacks theoretical guarantees on its span
and cache behavior.
\end{abstract}

\section{Introduction}

A \defn{parallel partition} operation rearranges the elements in an
array so that the elements satisfying a particular \defn{pivot
  property} appear first. In addition to playing a central role in
parallel quicksort, the parallel partition operation is used as a
primitive throughout parallel algorithms.\footnote{In several
  well-known textbooks and surveys on parallel algorithms
  \cite{AcarBl16,Blelloch96}, for example, parallel partitions are
  implicitly used extensively to perform what are referred to as
  \emph{filter} operations.}

A parallel algorithm can be measured by its \defn{work}, the time
needed to execute in serial, and its \defn{span}, the time to execute
on infinitely many processors. There is a well-known algorithm for
parallel partition on arrays of size $n$ that we call the
\defn{Standard Algorithm} with work $O(n)$ and span
$O(\log n)$ \cite{Blelloch96,AcarBl16}. Moreover, the algorithm uses
only exclusive read/write shared memory variables (i.e., it is an
\defn{EREW} algorithm). This eliminates the need for concurrency
mechanisms such as locks and atomic variables, and ensures good
behavior even if the time to access a location is a function of the
number of threads trying to access it (or its cache line)
concurrently. EREW algorithms also have the advantage that their
behavior is internally deterministic, meaning that
the behavior of the algorithm will not differ from run to run, which
makes test coverage, debugging, and reasoning about performance
substantially easier \cite{BlellochFi12}.

The Standard Algorithm suffers from using a large amount of
auxiliary memory, however. Whereas the serial partition algorithm
is typically implemented in place, the Standard Algorithm relies
on the use of multiple auxiliary arrays of size $n$. To the best
of our knowledge, the only known linear-work and
$\polylog(n)$-span algorithms for parallel partition that are
in-place require the use of atomic operations (e.g,
fetch-and-add) \cite{HeidelbergerNo90,AxtmannWi17,TsigasZh03}.

An algorithm's memory efficiency can be critical on large inputs. The
memory consumption of an algorithm determines the largest problem size
that can be executed in memory. Many external memory algorithms (i.e.,
algorithms for problems too large to fit in memory) perform large
subproblems in memory; the size of these subproblems is again
bottlenecked by the algorithm's memory-overhead \cite{Vitter08}. In
multi-user systems, processes with larger memory-footprints can hog
the cache and the memory bandwidth, slowing down other processes.

For sorting algorithms, in particular, special attention to memory
efficiency is often given. This is because (a) a user calling the sort
function may already be using almost all of the memory in the system;
and (b) sorting algorithms, and especially parallel sorting
algorithms, are often bottlenecked by memory bandwidth. The latter
property, in particular, means that any parallel sorting algorithm
that wishes to achieve state-of-the art performance on a large
multi-processor machine must be (at least close to) in place.

Currently the only practical in-place parallel sorting algorithms either rely
heavily on concurrency mechanisms such as atomic operations
\cite{HeidelbergerNo90, AxtmannWi17, TsigasZh03}, or abstain from
theoretical guarantees \cite{FrancisPa92}. Parallel merge sort
\cite{Hagerup89} was made in-place by Katajainen \cite{Katajainen93},
but has proven too sophisticated for practical applications. Bitonic
sort \cite{BlellochLe98} is naturally in-place, and can be practical
in certain applications on super computers, but suffers in general
from requiring work $\Theta(n \log^2 n)$ rather than $O(n \log
n)$. Parallel quicksort, on the other hand, despite the many efforts
to optimize it \cite{HeidelbergerNo90, AxtmannWi17, TsigasZh03,
  FrancisPa92, Frias08}, has eluded any in-place EREW (or CREW\footnote{In a
  \defn{CREW} algorithm, reads may be concurrent, but writes may
  not. CREW stands for \emph{concurrent-read exclusive-write}.})
algorithms due to its reliance on parallel partition.

\paragraph{Results}
We consider the problem of designing a theoretically efficient
parallel-partition algorithm that also performs well in practice. All
of the algorithms considered in this paper use only exclusive
read/write shared variables, and can be implemented using
CILK parallel-for-loops without any additional concurrency considerations.

We give a simple in-place parallel-partition algorithm called the
\defn{Smoothed Striding Algorithm}, with linear work and
polylogarithmic span. Additionally, the algorithm exhibits provably
optimal cache behavior up to low-order terms. In particular, if the
input consists of $m$ cache lines, then the algorithm incurs at most
$m(1 + o(1))$ cache misses, with high probability in $m$.

We also develop a suite of techniques for transforming the standard
linear-space parallel partition algorithm into an in-place
algorithm. The new algorithm, which we call the
\defn{Blocked-Prefix-Sum Partition Algorithm}, has work $O(n)$ and span
$O(\log n \cdot \log \log n)$, which is within a $\log \log n$ factor
of optimal. As an immediate consequence, we also get an in-place
quicksort algorithm with work $O(n \log n)$ and span
$O(\log^2 n \log \log n)$. Moreover, we show that, by using the
Blocked-Prefix-Sum Partition Algorithm as a subroutine within the
Smoothed Striding Algorithm, one can combine the theoretical
guarantees of the two algorithms, achieving a span of
$O(\log n \log \log n)$, while also achieving optimal cache behavior
up to low-order terms.

In addition to analyzing the algorithms, we experimentally evaluate
their performances in practice. We find that an algorithm based on the
Blocked-Prefix-Sum Partition Algorithm is able to achieve speedups over
the standard linear-space algorithm due to increased cache
efficiency. The Blocked-Prefix-Sum Partition Algorithm does not exhibit
optimal cache behavior, however, and as we show in our experimental
evaluation, the algorithm remains bottlenecked by memory
throughput. In contrast, the cache-optimality of the Smoothed Striding
Algorithm eliminates the memory-throughput bottleneck, allowing for
nearly perfect scaling on many processors.

The memory-bandwidth bottleneck previously led researchers
\cite{FrancisPa92, Frias08} to introduce the \defn{Strided
  Algorithm}, which has near optimal cache behavior in practice, but
which exhibits theoretical guarantees only on certain random input
arrays. The Smoothed Striding Algorithm is designed to have similar
empirical performance to the Strided Algorithm, while achieving both
theoretical guarantees on work/span and on cache-optimality. This is
achieved by randomly perturbing the internal structure of the Strided
Algorithm, and adding a recursion step that was previously not
possible. Whereas the Strided Algorithm comes with theoretical
guarantees only for certain inputs, the Smoothed Striding Algorithm 
has polylogarithmic span, and exhibits provably optimal cache behavior
up to small-order factors for all inputs. In practice, the Smoothed Striding
Algorithm performs within 15\% of the Strided Algorithm on a large
number of threads.

\paragraph{Outline} We begin in Section \ref{secprelim} by discussing
background on parallel algorithms and the Parallel Partition
Problem. In Section \ref{sec:smoothing}, we present and analyze
the Smoothed Striding Algorithm. In Section \ref{secalg} we
present and analyze the Blocked-Prefix-Sum Partition Algorithm,
which achieves a nearly optimal span of $O(\log n \log \log n)$
but is not cache-optimal; this algorithm can be used as a
subroutine within the Smoothed Striding Algorithm to achieve the
same span. Section \ref{secexp} implements the algorithms from
this paper, along with the Strided Algorithm and the standard
linear-space algorithm, in order to experimentally evaluate their
performances. Finally, we conclude with open questions in Section
\ref{sec:open}. 

\section{Preliminaries}\label{secprelim}

We begin by describing the parallelism and memory model used in
the paper, and by presenting background on the parallel partition problem.

\paragraph{Workflow Model} We consider a simple language-based model of
parallelism in which algorithms achieve parallelism through the use of
\defn{parallel-for-loops} (see, e.g.,
\cite{Blelloch96,AcarBl16,CLRS}); function calls within the inner loop
then allow for more complicated parallel structures (e.g., recursion). Our
algorithms can also be implemented in the PRAM model \cite{Blelloch96,
AcarBl16}.

Formally, a parallel-for-loop is given a range size $R \in \mathbb{N}$, a
constant number of arguments $\arg_1, \arg_2, \ldots, \arg_c$, and a
body of code. For each $i \in \{1, \ldots, R\}$, the loop launches a
thread that is given loop-counter $i$ and local copies of the
arguments $\arg_1, \arg_2, \ldots, \arg_c$. The threads are then taken up by
processors and the iterations of the loop are performed in parallel. Only after
every iteration of the loop is complete can control flow continue past the
loop.

A parallel algorithm may be run on an arbitrary number $p$ of
processors. The algorithm itself is oblivious to $p$, however, leaving
the assignment of threads to processors up to a scheduler.

The \defn{work} $T_1$ of an algorithm is the time that the algorithm
would require to execute on a single processor. The \defn{span}
$T_\infty$ of an algorithm is the time to execute on infinitely many
processors. The scheduler is assumed to contribute no overhead to the
span. In particular, if each iteration of a
parallel-for-loop has span $s$, then the full parallel loop has span
$s + O(1)$ \cite{Blelloch96,AcarBl16}.

The work $T_1$ and span $T_\infty$ can be used to quantify the time $T_p$
that an algorithm requires to execute on $p$ processors using a greedy
online scheduler. If the scheduler is assumed to contribute no
overhead, then Brent's Theorem \cite{Brent74} states that for any
$p$,
$$\max(T_1 / p, T_\infty) \le T_p \le T_1 / p + T_\infty.$$

The work-stealing algorithms used in the Cilk extension of C/C++ realize
the guarantee offered by Brent's Theorem within a constant factor
\cite{BlumofeJo96,BlumofeLe99}, with the added caveat that parallel-for-loops typically induce an additional additive overhead of $O(\log
R)$. 

\paragraph{Memory Model}
Memory is \defn{exclusive-read} and \defn{exclusive-write}. That is,
no two threads are ever permitted to attempt to read or write to the
same variable concurrently.  The exclusive-read exclusive-write memory
model is sometime referred to as the \defn{EREW model} (see, e.g.,
\cite{Hagerup89}). 

Note that threads are not in lockstep (i.e., they may progress at arbitrary
different speeds), and thus the EREW model requires algorithms to be data-race
free in order to avoid the possibility of non-exclusive data accesses.

In an \defn{in-place} algorithm, each thread is given $O(\polylog n)$
memory upon creation that is deallocated when the thread dies. This
memory can be shared with the thread's children. However, the depth of
the parent-child tree is not permitted to exceed $O(\polylog n)$.

Whereas the EREW memory model prohibits concurrent accesses to memory,
on the other side of the spectrum are CRCW
(concurrent-read-concurrent-write) models, which allow for both reads
and writes to be performed concurrently (and in some variants even
allow for atomic operations)
\cite{Blelloch96,AcarBl16,MatiasVi95}. One approach to designing
efficient EREW algorithms is to simulate efficient CRCW algorithms in
the EREW model \cite{MatiasVi95}. The known simulation techniques
require substantial space overhead, however, preventing the design of
in-place algorithms \cite{MatiasVi95}.\footnote{The known simulation
  techniques also increase the total work in the original algorithm,
  although this can be acceptable if only a small number of atomic
  operations need to be simulated.}

In addition to being the first in-place and polylogarithmic-span EREW
algorithms for the parallel-partition problem, our algorithms are also
the first such CREW algorithms. In a \defn{CREW} algorithm, reads may
be concurrent, but writes may not -- CREW stands for
\emph{concurrent-read exclusive-write}. In practice, the important
property of our algorithms is that they avoid concurrent writes (which
can lead to non-determinacy and cache ping-pong effects).



\paragraph{The Parallel Partition Problem}
The parallel partition problem takes an input array
$A = (A[1], A[2], \ldots, A[n])$ of size $n$, and a \defn{decider
  function} $\dec$ that determines for each element $A[i] \in A$
whether or not $A[i]$ is a \defn{predecessor} or a
\defn{successor}. That is, $\dec(A[i]) = 1$ if $A[i]$ is a
predecessor, and $\dec(A[i]) = 0$ if $A[i]$ is a successor. The
behavior of the parallel partition is to reorder the elements in the
array $A$ so that the predecessors appear before the successors. Note
that, in this paper, we will always treat arrays as 1-indexed.

\paragraph{The (Standard) Linear-Space Parallel Partition} The linear-space
implementation of parallel partition consists of two phases
\cite{Blelloch96,AcarBl16}:

\noindent\emph{The Parallel-Prefix Phase: }In this phase, the algorithm first
creates an array $D$ whose $i$-th element $D[i] = \dec(A[i])$. Then the
algorithm constructs an array $S$ whose $i$-th element $S[i] = \sum_{j = 1}^i
D[i]$ is the number of predecessors in the first $i$ elements of $A$. The
transformation from $D$ to $S$ is called a \defn{parallel prefix sum} and can
be performed with $O(n)$ work and $O(\log n)$ span using a simple recursive
algorithm: (1) First construct an array $D'$ of size $n / 2$ with $D'[i] = D[2i
- 1] + D[2i]$; (2) Recursively construct a parallel prefix sum $S'$ of $D'$;
(3) Build $S$ by setting each $S[i] = S'[\lfloor i / 2 \rfloor] + A[i]$ for odd
$i$ and $S[i] = S'[i / 2]$ for even $i$. 

\noindent\emph{The Reordering Phase: }In this phase, the algorithm constructs
an output-array $C$ by placing each predecessor $A[i] \in A$ in position $S[i]$
of $C$. If there are $t$ predecessors in $A$, then the first $t$ elements of
$C$ will now contain those $t$ predecessors in the same order that they appear
in $A$. The algorithm then places each successor $A[i] \in A$ in position $t +
i - S[i]$. Since $i - S[i]$ is the number of successors in the first $i$
elements of $A$, this places the successors in $C$ in the same order that they
appear in $A$. Finally, the algorithm copies $C$ into $A$, completing the
parallel partition.

Both phases can be implemented with $O(n)$ work and $O(\log n)$
span. Like its serial out-of-place counterpart, the algorithm is
stable but not in place. The algorithm uses multiple auxiliary arrays of
size $n$. Kiu, Knowles, and Davis \cite{LiuKn05} were able to reduce
the extra space consumption to $n + p$ under the assumption that the
number of processors $p$ is hard-coded; their algorithm breaks the
array $A$ into $p$ parts and assigns one part to each thread. Reducing
the extra space below $o(n)$ has remained open until now, even when the
number of threads is fixed.

\section{An In-Place Partition Algorithm with Span $O(\log n \log \log n)$}\label{secalg}

In this section, we present an in-place algorithm for parallel
partition with span $O(\log n \log \log n)$. Each thread in the
algorithm requires memory at most $O(\log n)$.

Prior to beginning the algorithm, the first implicit step of the
algorithm is to count the number of predecessors in the array, in
order to determine whether the majority of elements are either
predecessors or successors. Throughout the rest of the section, we
assume without loss of generality that the total number of successors
in $A$ exceeds the number of predecessors, since otherwise their roles
can simply be swapped in the algorithm. Further, we assume for
simplicity that the elements of $A$ are distinct; this assumption is
removed at the end of the section.

\paragraph{Algorithm Outline}
We begin by presenting an overview of the key algorithmic ideas needed
to construct an in-place algorithm.

Consider how to remove the auxiliary array $C$ from the Reordering
Phase. If one attempts to simply swap in parallel each predecessor
$A[i]$ with the element in position $j = S[i]$ of $A$, then the swaps
will almost certainly conflict. Indeed, $A[j]$ may also be a
predecessor that needs to be swapped with $A[S[j]]$. Continuing like
this, there may be an arbitrarily long list of dependencies on the
swaps.

To combat this, we begin the algorithm with a Preprocessing Phase in
which $A$ is rearranged so that every prefix is
\defn{successor-heavy}, meaning that for all $t$, the first $t$
elements contain at least $\frac{t}{4}$ successors. Then we compute
the prefix-sum array $S$, and begin the Reordering Phase. Using the
fact that the prefixes of $A$ are successor-heavy, the reordering can
now be performed in place as follows: (1) We begin by recursively
reordering the prefix $P$ of $A$ consisting of the first $4/5 \cdot n$
elements, so that the predecessors appear before the successors; (2)
Then we simply swap each predecessor $A[i]$ in the final $1/5 \cdot n$
elements with the corresponding element $S[A[i]]$. The fact that the
prefix $P$ is successor-heavy ensures that, after step (1), the final 
$\frac{1}{5} \cdot n$ elements of (the reordered) $P$ are successors. 
This implies in step (2) that for each of the swaps between predecessors $A[i]$
in the final $1/5 \cdot n$ elements and earlier positions $S[A[i]]$, the latter
element will be in the prefix $P$. In other words, the swaps are now conflict
free.

Next consider how to remove the array $S$ from the Parallel-Prefix
Phase. At face value, this would seem quite difficult since the
reordering phase relies heavily on $S$. Our solution is to
\emph{implicitly} store the value of every $O(\log n)$-th element of
$S$ in the ordering of the elements of $A$. That is, we break $A$ into
blocks of size $O(\log n)$, and use the order of the elements in each
block to encode an entry of $S$. (If the elements are not all
  distinct, then a slightly more sophisticated encoding is necessary.)
Moreover, we modify the algorithm for building $S$ to only construct
every $O(\log n)$-th element. The new parallel-prefix sum performs
$O(n / \log n)$ arithmetic operations on values that are implicitly
encoded in blocks; since each such operation requires $O(\log n)$
work, the total work remains linear.

In the remainder of the section, we present the algorithm in
detail, and prove the key properties of each phase of the
algorithm. We also provide detailed pseudocode in Figure
\ref{alg:parallelPartition_prefixsumbased_helpers} and Figure
\ref{alg:parallelPartition_prefixsumbased_main}. The algorithm
proceeds in three phases.

\begin{figure*}
  \caption{Blocked-Prefix-Sum Partition Algorithm: Helper Functions}
  \label{alg:parallelPartition_prefixsumbased_helpers}
  \begin{algorithmic} 
    \Procedure{WriteToBlock}{$A$, $b$, $i$, $v$} \Comment Write value $v$ to the $i$-th block $X_i$ of $A$, where $A = X_1 \circ X_2 \circ \cdots \circ X_{\lfloor n/b \rfloor}$
      \ForAll{$j \in \{1,2,\ldots, \lfloor b/2 \rfloor \}$ in parallel}
      \If{$\mathds{1}_{X_i[2j] < X_i[2j+1]} \not= $ (the $j$-th binary digit of $v$)}
          \State Swap $X_i[2j]$ and $X_i[2j+1]$
        \EndIf
      \EndFor
    \EndProcedure
    \State

    \Procedure{ReadFromBlock}{$A$, $i$, $j$} \Comment Reads the value $v$ stored in $A[i], A[i+1], \ldots, A[j]$
      \If{$j-i=2$}
        \State \Return $\mathds{1}_{A[i] < A[i+1]}$
      \Else
        \State \emph{Parallel-Spawn} $v_0 \gets $ ReadFromBlock($A$, $i$, $i+(j-i)/2$)
        \State \emph{Parallel-Spawn} $v_f \gets $ ReadFromBlock($A$, $i+(j-i)/2+1$, $j$)
        \State \emph{Parallel-Sync}
        \State \Return $v_f\cdot 2^\frac{j-i}{4} + v_0$
      \EndIf
    \EndProcedure
    \State

    \State \textbf{Require: } $A$ has more successors than predecessors
    \State \textbf{Ensure: }  Each prefix of $A$ is ``successor heavy''
    \Procedure{MakeSuccessorHeavy}{$A$, $n$}
      \ForAll{$i \in \{1,2,\ldots,\lfloor n/2 \rfloor\}$ in parallel}
        \If{$A[i]$ is a predecessor and $A[n-i+1]$ is a successor}
          \State Swap $A[i]$ and $A[n-i+1]$
        \EndIf
      \EndFor
      \State MakeSuccessorHeavy($A$, $\lceil n/2 \rceil$)
      \Comment Recurse on $A[1],A[2], \ldots, A[\lceil n/2 \rceil]$
    \EndProcedure
  \end{algorithmic}	
\end{figure*}

\begin{figure*}
  \caption{Blocked-Prefix-Sum Partition Algorithm: Main Functions}
  \label{alg:parallelPartition_prefixsumbased_main}
  \begin{algorithmic} 
    \State \textbf{Require: } Each prefix of $A$ is ``successor heavy''
    \State \textbf{Ensure: }  Each block $X_i$ stores how many predecessors occur in $X_1 \circ X_2 \circ \cdots \circ X_i$
    \Procedure{ImplicitParallelPrefixSum}{A, n}
      \State Pick $b \in \Theta(\log n)$ to be the "block size"
      \State Logically Partition $A$ into blocks, with $A = X_1 \circ X_2 \circ \cdots \circ X_{\lfloor n/b \rfloor}$
      \ForAll{$i \in \{1,2,\ldots,\lfloor n/b \rfloor\}$ in parallel}
        \State $v_i \gets 0$ \Comment{$v_i$ will store number of predecessors in $X_i$ } 
        \ForAll{$a \in X_i$ in serial}
          \If{$a$ is a predecessor}
            \State $v_i \gets v_i + 1$
          \EndIf
        \EndFor
        \State WriteToBlock($A$, $b$, $i$, $v_i$)
        \Comment Now we encode the value $v_i$ in the block $X_i$
      \EndFor
      \State Perform a parallel prefix sum on the values $v_i$ stored in the $X_i$'s
   \EndProcedure
    \State

  \State \textbf{Require: } Each block $X_i$ stores how many predecessors occur in $X_1 \circ X_2 \circ \cdots \circ X_i$
  \State \textbf{Ensure: }  $A$ is partitioned
    \Procedure{Reorder}{$A$, $n$}
      \State $t \gets $ least integer such that $t\cdot b > n\cdot 4/5$
      \State Reorder($A$, $t$)
      \Comment Recurse on $A[1], A[2], \ldots, A[t]$
      \ForAll{$i \in \{t+1, t+2, \ldots, \lfloor n/b \rfloor\}$}
      \State $v_i \gets$ ReadFromBlock($A$, $b\cdot i+1$, $b\cdot(i+1)$) 
        \State Instantiate an array $Y_i$ with $|Y_i| = |X_i| \in \Theta(\log n)$, 
        \State In parallel, set $Y_i[j] \gets 1$ if $X_i[j]$ is a predecessor, and $Y_i[j] \gets 0$ otherwise.
        \State Perform a parallel prefix sum on $Y_i$, and add $v_i$ to each $Y_i[j]$
        \ForAll{$j \in \{1,2,\ldots, b\}$}
          \If{$X_i[j]$ is a predecessor}
            \State Swap $X_i[j]$ and $A[Y_i[j]]$
          \EndIf
        \EndFor
      \EndFor
    \EndProcedure
    \State

    \Procedure{ParallelPartition}{$A$, $n$}
      \State $k \gets$ count number of successors in $A$ in parallel
      \If{$k < n/2$}
        \State Swap the role of successors and predecessors in the algorithm (i.e. change the decider function)
        \State At the end we consider $A'[i] = A[n-i+1]$, the logically reversed array, as output
      \EndIf

      \State MakeSuccessorHeavy($A$, $n$) \Comment \emph{prepreocessing phase}
      \State ImplicitParallelPrefixSum($A$, $n$) \Comment \emph{Implicit Parallel Prefix Sum}
      \State Reorder($A$, $n$) \Comment \emph{In-Place Reordering Phase}
    \EndProcedure
  \end{algorithmic}	
\end{figure*}

\paragraph{A Preprocessing Phase}
The goal of the Preprocessing phase is to make every prefix of $A$
successor-heavy. To perform the Preprocessing phase on $A$, we begin
with a parallel-for-loop: For each $i = 1, \ldots, \lfloor n /
2\rfloor$, if $A[i]$ is a predecessor and $A[n - i + 1]$ is a
successor, then we swap their positions in $A$. To complete the
Preprocessing phase on $A$, we then recursively perform a
Preprocessing phase on $A[1], \ldots, A[\lceil n / 2 \rceil]$.

\begin{lemma}
 The Preprocessing Phase has work $O(n)$ and span $O(\log n)$. At the
 end of the Preprocessing Phase, every prefix of $A$ is
 successor-heavy.
  \label{lem:preprocessingphase}
\end{lemma}
\begin{proof}
Recall that for each $t \in 1, \ldots, n$, we call the $t$-prefix
$A[1], \ldots, A[t]$ of $A$ successor-heavy if it contains at least
$\frac{t}{4}$ successors.

The first parallel-for-loop ensures that at least half the successors
in $A$ reside in the first $\lceil n / 2 \rceil$ positions, since for
$i = 1, \ldots, \lfloor n / 2 \rfloor$, $A[n - i + 1]$ will only be a
successor if $A[i]$ is also a successor. Because at least half the
elements in $A$ are successors, it follows that the first $\lceil n /
2 \rceil$ positions contain at least $\lceil n / 4\rceil$ successors,
making every $t$-prefix with $t \ge \lceil n / 2 \rceil$
successor-heavy.

After the parallel-for-loop, the first $\lceil n / 2 \rceil$ positions
of $A$ contain at least as many successors as predecessors (since
$\lceil n / 4 \rceil \ge \frac{\lceil n / 2 \rceil}{2}$). Thus we can
recursively apply the argument above in order to conclude that the
recursion on $A[1], \ldots, A[\lceil n / 2 \rceil]$ makes every
$t$-prefix with $t \le \lceil n / 2 \rceil$ successor-heavy. It
follows that, after the recursion, every $t$-prefix of $A$ is
successor-heavy.

Each recursive level has constant span and performs work proportional
to the size of the subarray being considered. The Preprocessing phase
therefore has total work $O(n)$ and span $O(\log n)$.
\end{proof}






\paragraph{An Implicit Parallel Prefix Sum}
Pick a \defn{block-size} $b \in \Theta(\log n)$ satisfying $b \ge 2
\lceil \log (n + 1) \rceil$. Consider $A$ as a series of $\lfloor n /
b \rfloor$ blocks of size $b$, with the final block of size between
$b$ and $2b - 1$. Denote the blocks by $X_1, \ldots, X_{\lfloor n / b
  \rfloor}$.

Within each block $X_i$, we can implicitly store a value in the range
$0, \ldots, n$ through the ordering of the elements:
\begin{lemma}
Given an array $X$ of $2 \lceil \log (n + 1) \rceil$ distinct
elements, and a value $v \in \{0, \ldots, n\}$, one can rearrange the
elements of $X$ to encode the bits of $v$ using work $O(\log n)$ and
span $O(\log \log n)$; and one can then later decode $v$ from $X$
using work $O(\log n)$ and span $O(\log \log n)$.
\label{lem:bitstore}
\end{lemma}
\begin{proof}
Observe that $X$ can be broken into (at least) $\lceil \log (n + 1)
\rceil$ disjoint pairs of adjacent elements $(x_1, x_2), (x_3, x_4),
\ldots$, and by rearranging the order in which a given pair $(x_j,
x_{j + 1})$ occurs, the lexicographic comparison of whether $x_j <
x_{j + 1}$ can be used to encode one bit of information. Values $v \in
[0,n]$ can therefore be read and written to $X$ with work $O(b) =
O(\log n)$ and span $O(\log b) = O(\log \log n)$ using a simple
divide-and-conquer recursive approach to encode and decode the bits of
$v$.
\end{proof}

To perform the Parallel Prefix Sum phase, our algorithm begins by
performing a parallel-for loop through the blocks, and storing in each
block $X_i$ a value $v_i$ equal to the number of predecessors in the
block. (This can be done in place with work $O(n)$ and span $O(\log
\log n)$ by Lemma \ref{lem:bitstore}.)

The algorithm then performs an in-place parallel-prefix operation on
the values $v_1, \ldots, v_{\lfloor n / b \rfloor}$ stored in the
blocks. This is done by first resetting each even-indexed value
$v_{2i}$ to $v_{2i} + v_{2i - 1}$; then recursively performing a
parallel-prefix sum on the even-indexed values; and then replacing
each odd-indexed $v_{2i + 1}$ with $v_{2i + 1} + v_{2i}$, where $v_0$
is defined to be zero.

Lemma \ref{lem:parallelprefix} analyzes the phase:
\begin{lemma}
  The Parallel Prefix Sum phase uses work $O(n)$ and span
  $O(\log n \log \log n)$. At the end of the phase, each $X_i$ encodes
  a value $v_i$ counting the number of predecessors in the prefix
  $X_1 \circ X_2 \circ \cdots \circ X_i$; and each prefix of blocks
  (i.e., each prefix of the form
  $X_1 \circ X_2 \circ \cdots \circ X_i$) is successor-heavy.
\label{lem:parallelprefix}
\end{lemma}
\begin{proof}
If the $v_i$'s could be read and written in constant time, then the
prefix sum would take work $O(n / \log n)$ and span $O(\log n)$, since
there are $O(n / \log n)$ $v_i$'s. Because each $v_i$ actually
requires work $O(\log n)$ and span $O(\log \log n)$ to read/write (by
Lemma \ref{lem:bitstore}), the prefix sum takes work $O(n)$ and span
$O(\log n \cdot \log \log n)$.

Once the prefix-sum has been performed, every block $X_i$
encodes a value $v_i$ counting the number of predecessors in the
prefix $X_1 \circ X_2 \circ \cdots \circ X_i$. Moreover, because the
Parallel Prefix Sum phase only rearranges elements within each $X_i$,
Lemma \ref{lem:preprocessingphase} ensures that each prefix of the
form $X_1 \circ X_2 \circ \cdots \circ X_i$ remains successor-heavy.
\end{proof}

\paragraph{In-Place Reordering}
In the final phase of the algorithm, we reorder $A$ so that the
predecessors appear before the successors. Let $P = X_1 \circ X_2
\circ \cdots \circ X_t$ be the smallest prefix of blocks that contains
at least $4/5$ of the elements in $A$. We begin by recursively
reordering the elements in $P$ so that the predecessors appear before
the successors; as a base case, when $|P| \le 5b = O(\log n)$, we
simply perform the reordering in serial.

To complete the reordering of $A$, we perform a parallel-for-loop
through each of the blocks $X_{t + 1}, \ldots, X_{\lfloor n / b \rfloor}$. For each block
$X_i$, we first extract $v_i$ (with work $O(\log n)$ and span $O(\log
\log n)$ using Lemma \ref{lem:bitstore}). We then create an auxiliary
array $Y_i$ of size $|X_i|$, using $O(\log n)$ thread-local
memory. Using a parallel-prefix sum (with work $O(\log n)$ and span
$O(\log \log n)$), we set each $Y_i[j]$ equal to $v_i$ plus the number
of predecessors in $X_i[1], \ldots, X_i[j]$. In other words, $Y_i[j]$
equals the number of predecessors in $A$ appearing at or before
$X_i[j]$.

After creating $Y_i$, we then perform a parallel-for-loop through the
elements $X_i[j]$ of $X_i$ (note we are still within another parallel
loop through the $X_i$'s), and for each predecessor $X_i[j]$, we swap
it with the element in position $Y_i[j]$ of the array $A$. This
completes the algorithm.

\begin{lemma}
 The Reordering phase takes work $O(n)$ and span $O(\log n \log \log
 n)$. At the end of the phase, the array $A$ is fully partitioned.
\end{lemma}
\begin{proof}
  After $P$ has been recursively partitioned, it will be of the form
  $P_1 \circ P_2$ where $P_1$ contains only predecessors and $P_2$
  contains only successors. Because $P$ was successor-heavy before the
  recursive partitioning (by Lemma \ref{lem:parallelprefix}), we have
  that $|P_2| \ge |P| / 4$, and thus that
  $|P_2| \ge |X_{t + 1} \circ \cdots \circ X_{\lfloor n / b
    \rfloor}|$.

After the recursion, the swaps performed by the algorithm will swap
the $i$-th predecessor in $X_{t + 1} \circ \cdots \circ X_{\lfloor n /
  b \rfloor}$ with the $i$-th element in $P_2$, for $i$ from $1$ to
the number of predecessors in $X_{t + 1} \circ \cdots \circ X_{\lfloor
  n / b \rfloor}$. Because $|P_2| \ge |X_{t + 1} \circ \cdots \circ
X_{\lfloor n / b \rfloor}|$ these swaps are guaranteed not to conflict
with one-another; and since $P_2$ consists of successors, the final
state of array $A$ will be fully partitioned.

The total work in the reordering phase is $O(n)$ since each $X_i$
appears in a parallel-for-loop at exactly one level of the recursion,
and incurs $O(\log n)$ work. The total span of the reordering phase is
$O(\log n \cdot \log \log n)$, since there are $O(\log n)$ levels of
recursion, and within each level of recursion each $X_i$ in the
parallel-for-loop incurs span $O(\log \log n)$. 
\end{proof}

Combining the phases, the full algorithm has work $O(n)$ and span
$O(\log \log n)$. Thus we have:
\begin{theorem}
  There exists an in-place algorithm using exclusive-read-write
  variables that performs parallel-partition with work $O(n)$ and span
  $O(\log n \cdot \log \log n)$.
  \label{thminplace}
\end{theorem}

\paragraph{Allowing for Repeated Elements}
In proving Theorem \ref{thminplace} we assumed for simplicity that the
elements of $A$ are distinct. To remove this assumption, we conclude
the section by proving a slightly more complex variant of Lemma
\ref{lem:bitstore}, eliminating the requirement that the elements of
the array $X$ be distinct:

\begin{lemma}
Let $X$ be an array of $b = 4 \lceil \log (n + 1) \rceil + 2$
elements. The there is an \emph{encode} function, and a \emph{decode}
function such that:
\begin{itemize}
\item The encode function modifies the array $X$ (possibly overwriting
  elements in addition to rearranging them) to store a value $v \in
  \{0, \ldots, n\}$. The first time the encode function is called on
  $X$ it has work and span $O(\log n)$. Any later times the encode
  function is called on $X$, it has work $O(\log n)$ and span $O(\log
  \log n)$. In addition to being given an argument $v$, the encode
  function is given a boolean argument indicating whether the function
  has been invoked on $X$ before.
\item The decode function recovers the value of $v$ from the modified
  array $X$, and restores $X$ to again be an array consisting of the
  same multiset of elements that it began with. The decode function
  has work $O(\log n)$ and span $O(\log \log n)$.
\end{itemize}
  \label{lem:bitstore2}
\end{lemma}
\begin{proof}
Consider the first $b$ letters of $X$ as a sequence of pairs, given by
$(x_1, x_2), \ldots, (x_{b - 1}, x_b)$. If at least half of the pairs
$(x_i, x_{i + 1}$ satisfy $x_i \neq x_{i + 1}$, then the encode
function can reorder those pairs to appear at the front of $X$, and
then use them to encode $v$ as in Lemma \ref{lem:bitstore}. Note that
the reordering of the pairs will only be performed the first time that
the encode function is invoked on $X$. Later calls to the encode
function will have work $O(\log n)$ and span $O(\log \log n)$, as in
Lemma \ref{lem:bitstore}.

If, on the other hand, at least half the pairs consist of equal-value
elements $x_i = x_{i + 1}$, then we can reorder the pairs so that the
first $\lceil \log (n + 1) \rceil + 1$ of them satisfy this
property. (This is only done on the first call to encode.) To encode a
value $v$, we simply explicitly overwrite the second element in each
of the pairs $(x_3, x_4), (x_5, x_6), \ldots$ with the bits of $v$,
overwriting each element with one bit. The reordering performed by the
first call to encode has work and span $O(\log n)$; the writing of
$v$'s bits can then be performed in work $O(\log n)$ and span $O(\log
\log n)$ using a simple divide-and-conquer approach.

To perform a decode and read the value $v$, we check whether $x_1 =
x_2$ in order to determine which type of encoding is being used, and
then we can unencode the bits of $v$ using work $O(\log n)$ and span
$O(\log \log n)$; if the encoding is the second type (i.e., $x_1 =
x_2$), then the decode function also restores the elements $x_2, x_4,
x_6, \ldots$ of the array $X$ as it extracts the bits of $v$. Note
that checking whether $x_1 = x_2$ is also used by the encode function
each time after the first time it is called, in order determine which
type of encoding is being used.
\end{proof}

The fact that the first call to the encode function on each $X_i$ has
span $O(\log n)$ (rather than $O(\log \log n)$) does not affect the
total span of our parallel-partition algorithm, since this simply adds
a step with $O(\log n)$-span to the beginning of the Parallel Prefix
phase. Lemma \ref{lem:bitstore2} can therefore used in place of Lemma
\ref{lem:bitstore} in order to complete the proof of Theorem
\ref{thminplace} for arrays $A$ that contain duplicate elements.

\section{A Cache-Efficient Partition Algorithm}\label{sec:smoothing}


In this section we present the \defn{Smoothed Striding Algorithm},
which exhibits provably optimal cache behavior (up to small-order
factors). The Smoothed Striding Algorithm is fully in-place and has
polylogarithmic span. In particular, this means that the total amount
of auxiliary memory allocated at a given moment in the execution never
exceeds $\polylog n$ per active worker.


\paragraph{Modeling Cache Misses}
We treat memory as consisting of fixed-size cache lines, each of some
size $b$. Each processor is assumed to have a small cache of
$\operatorname{polylog}{n}$ cache lines.  A cache miss occurs on a
processor when the line being accessed is not currently in cache, in
which case some other line is evicted from cache to make room for the
new entry.  Each cache is managed with a LRU (Least Recently Used)
eviction policy.

We assume that threads are scheduled using work stealing \cite{AcarBl00},
and that the work-stealing itself has no cache-miss overhead. Note
that caches belong to processors, not threads, meaning that when a
processor takes a new thread (i.e., performs work stealing), the
processor's cache contents are a function of what the processor was
previously executing. In order to keep our analysis of cache misses
independent of the number of processors, we will ignore the cost of
warming up each processor's cache. In particular, if there are
$\polylog n$ global variables that each processor must access many
times, we do not consider the initial $p \cdot \polylog n$ cost of
loading those global variables into the caches (where $p$ is the
number of processors). In practice, $p \ll n$ on large inputs, making
the cost of warming up caches negligible.

Although each cache is managed with LRU eviction, we may assume
for the sake of analysis that each cache is managed by the optimal
off-line eviction strategy OPT (i.e. Furthest in the Future). This is
because, up to resource augmentation, LRU eviction is $(1 + 1/\polylog
n)$-competitive with OPT. Formally this is due to the following
theorem by Sleator and Tarjan \cite{SleatorTa85}:
\begin{theorem}
  LRU operating on a cache of size $K\cdot M$ for some $K>1$ will incur at most
  $1+\frac{1}{K-1}$ times the number of times cache misses of OPT operating on
  a cache of size $M$, for the same series of memory accesses.
  \label{thm:augmentation}
\end{theorem}

Recall that each processor has a cache of size $\log^c n$ for $c$ a
constant of our choice.  Up to changes in $c$ LRU incurs no more than
a $1+\frac{1}{\polylog{n}}$ factor more cache misses
than OPT incurs. Thus, up to a $1 + \frac{1}{\polylog(n)}$
multiplicative change in cache misses, and a $\polylog(n)$ factor change in
cache size, we may assume without loss of generality that cache
eviction is performed by OPT.

Because each processor's cache is managed by OPT (without loss of
generality), we can assume that each processor \defn{pins} certain
small arrays to cache (i.e., the elements of those arrays are never
evicted). In fact, this is the only property of OPT we will use; that
is, our analyses will treat non-pinned contents of the cache as being
managed via LRU.

\paragraph{The Strided Algorithm \cite{FrancisPa92}}
The Smoothed Striding Algorithm borrows several structural ideas
from a previous algorithm of Francis and Pannan
\cite{FrancisPa92}, which we call the Strided Algorithm. The
Strided Algorithm is designed to behave well on random arrays
$A$, achieving span $\tilde{O}(n^{2/3})$ and exhibiting only $n/b
+ \tilde{O}(n^{2/3} / b)$  cache misses on such inputs. On
worst-case inputs, however, the Strided Algorithm has span
$\Omega(n)$ and incurs $n/b + \Omega(n/b)$ cache misses. Our
algorithm, the Smoothed Striding Algorithm, builds on the Strided
Algorithm by randomly perturbing the internal structure of the
original algorithm; in doing so, we are able to provide provable
performance guarantees for arbitrary inputs, and to add a
recursion step that was previously impossible.

The \defn{Strided Algorithm} consists of two steps: \\
\textbf{The Partial Partition Step.} \\
Let $g \in \mathbb{N}$ be a parameter, and assume for simplicity
that $gb \mid n$. Logically partition the array $A$ into $n/(gb)$
chunks $C_1, \ldots, C_{n / (gb)}$, each consisting of $g$ cache
lines of size $b$. For $i \in \{1, 2, \ldots, g\}$, define $P_i$
to consist of the $i$-th cache line from each of the chunks $C_1,
\ldots, C_{n / (gb)}$. One can think of the $P_i$'s as forming a
strided partition of array $A$, since consecutive cache lines in
$P_i$ are always separated by a fixed stride of $g - 1$ other
cache lines.\\
The first step of the Strided Algorithm is to perform an in-place
serial partition on each of the $P_i$'s, rearranging the
elements within the $P_i$ so that the predecessors come first.
This step requires work $\Theta(n)$ and span $\Theta(n/g)$.\\
\textbf{The Serial Cleanup Step. }For each $P_i$, define
the \defn{splitting position} $v_i$ to be the position in $A$
of the first successor in (the already partitioned) $P_i$.
Define $v_{\text{min}} = \min\{v_1, \ldots, v_{g}\}$ and define
$v_{\text{max}} = \max\{v_1, \ldots, v_{g}\}$. \\
The second step of the Strided Algorithm is to perform a serial partition on the
sub-array $A[v_{\text{min}}],\ldots, A[v_{\text{max}}-1]$. This
step has no parallelism, and thus has work and span of $\Theta(v_{\text{max}} -
v_{\text{min}})$. 

In general, the lack of parallelism in the Serial Cleanup step
results in an algorithm with linear-span (i.e., no parallelism
guarantee).  When the number of predecessors in each of the
$P_i$'s is close to equal, however, the quantity $v_{\text{max}}
- v_{\text{min}}$ can be much smaller than $\Theta(n)$. For
example, if $b = 1$, and if each element of $A$ is selected
independently from some distribution, then one can use Chernoff
bounds to prove that with high probability in $n$,
$v_{\text{max}} - v_{\text{min}} \le O(\sqrt{n \cdot g \cdot \log
n})$.  The full span of the algorithm is then $\tilde{O}(n/g +
\sqrt{n \cdot g})$, which optimizes at $g = n^{1/3}$ to
$\tilde{O}(n^{2/3})$. Since the Partial Partition Step incurs
only $n / b$ cache misses, the full algorithm incurs $n +
\tilde{O}(n^{2/3})$ cache misses on a random array $A$.

Using Hoeffding's Inequality in place of Chernoff bounds, one can
obtain analogous bounds for larger values of $b$; in particular for $b
\in \operatorname{polylog}(n)$, the optimal span remains
$\tilde{O}(n^{2/3})$ and the number of cache misses becomes $n / b +
\tilde{O}(n^{2/3} / b)$ on an array $A$ consisting of randomly sampled
elements.\footnote{The original algorithm of Francis and Pannan
  \cite{FrancisPa92} does not consider the cache-line size $b$. Frias
  and Petit later introduced the parameter $b$ \cite{Frias08}, and
  showed that by setting $b$ appropriately, one obtains an algorithm
  whose empirical performance is close to the state-of-the-art.}


\paragraph{The Smoothed Striding Algorithm}
To obtain an algorithm with provable guarantees for all inputs $A$, we
randomly perturb the internal structure of each of the $P_i$'s. Define
$U_1, \ldots, U_{g}$ (which play a role analogous to $P_1,
\ldots, P_g$ in the Strided Algorithm) so that each $U_i$ contains one
randomly selected cache line from each of $C_1, \ldots, C_{n /
  gb}$ (rather than containing the $i$-th cache line of each
$C_j$). This ensures that the number of predecessors in each $U_i$ is
a sum of independent random variables with values in $\{0, 1, \ldots,
n/g\}$.

By Hoeffding's Inequality, with high probability in $n$, the number of
predecessors in each $U_i$ is tightly concentrated around $\frac{\mu
  n}{g}$, where $\mu$ is the fraction of elements in $A$ that are
predecessors. It follows that, if we perform in-place partitions of
each $U_i$ in parallel, and then define $v_i$ to be the position in
$A$ of the first successor in (the already partitioned) $U_i$, then
the difference between $v_{\text{min}} = \min_i v_i$ and
$v_{\text{max}} = \max_i v_i$ will be small (regardless of the input array
$A$!).

Rather than partitioning $A[v_{\text{min}}],\ldots,
A[v_{\text{max}}-1]$ in serial, the Smoothed Striding Algorithm simply
recurses on the subarray. Such a recursion would not have been
productive for the original Strided Algorithm because the strided
partition $P_1', \ldots, P_g'$ used in the recursive subproblem would
satisfy $P_1' \subseteq P_1, \ldots, P_g' \subseteq P_g$ and thus each
$P_i'$ is already partitioned. That is, in the original Strided
Algorithm, the problem that we would recurse on is a worst-case input
for the algorithm in the sense that the partial partition step makes
no progress.

The main challenge in designing the Smoothed Striding Algorithm
becomes the construction of $U_1, \ldots, U_{g}$ without
violating the in-place nature of the algorithm. A natural
approach might be to store for each $U_i, C_j$ the index of the
cache line in $C_j$ that $U_i$ contains. This would require the
storage of $\Theta(n / b)$ numbers as metadata, however,
preventing the algorithm from being in-place. To save space, the
key insight is to select a random offset $X_j \in \{1, 2, \ldots,
g\}$ within each $C_j$, and then to assign the $(X_j + i \pmod g)
+ 1$-th cache line of $C_j$ to $U_i$ for $i \in \{1, 2, \ldots,
g\}$. This allows for us to construct the $U_i$'s using only
$O(n/(gb))$ machine words storing the metadata $X_1, \ldots, X_{n
/ (gb)}$. By setting $g$ to be relatively large, so that $n/(gb)
\le \polylog(n)$, we can obtain an in-place algorithm that incurs
$(1 + o(1))n/b$ cache misses.

The recursive structure of the Smoothed Striding Algorithm allows for
the algorithm to achieve polylogarithmic span. As an alternative
to recursing, one can also use the Blocked-Prefix-Sum Partition
Algorithm from Section \ref{secalg} in order to partition
$A[v_{\text{min}}], \ldots, A[v_{\text{max}} - 1]$. This results
in an improved span (since the algorithm from Section
\ref{secalg} has span only $O(\log n \log \log n)$), while still
incurring only $(1 + o(1))n/b$ cache misses (since the
cache-inefficient algorithm from Section \ref{secalg} is only
used on a small subarray of $A$). We analyze both the recursive
version of the Smoothed Striding Algorithm, and the version which
uses as a final step the Blocked-Prefix-Sum Partition Algorithm; one
significant advantage of the recursive version is that it is
simple to implement in practice.

\paragraph{Formal Algorithm Description} Let $b < n$ be the size
of a cache line, let $A$ be an input array of size $n$, and let
$g$ be a parameter. (One should think of $g$ as being relatively
large, satisfying $n/(gb) \le \polylog(n)$.) We assume for
simplicity that that $n$ is divisible by $gb$, and we define $s =
n/(gb)$. \footnote{This assumption can be made without loss of
  generality by treating $A$ as an array of size $n' = n + {(gb -
  n \pmod {gb})}$, and then treating the final $gb - n \pmod
  {gb}$ elements of the array as being successors (which
  consequently the algorithm needs not explicitly access). Note
  that the extra $n' - n$ elements are completely virtual,
  meaning they do not physically exist or reside in memory.


}

In the \defn{Partial Partition Step} the algorithm partitions the
cache lines of $A$ into $g$ sets $U_1, \ldots, U_{g}$ where each
$U_i$ contains $s$ cache lines, and then performs a serial
partition on each $U_i$ in parallel over the $U_i$'s. To
determine the sets $U_1, \ldots, U_{g}$, the algorithm uses as
metadata an array $X = X[1], \ldots, X[s]$, where each $X[i] \in
\{1, \ldots, g\}$. More specifically the algorithm does the
following:

Set each of $X[1], \ldots, X[s]$ to be uniformly random and
independently selected elements of $\{1, 2, \ldots, g\}$. For each $i \in
\{1, 2, \ldots, g\}$, $j \in \{1, 2, \ldots, s\}$, define
$$G_i(j) = (X[j] + i \pmod g) + (j - 1)g + 1.$$ Using this
terminology, we define each $U_i$ for $i \in \{1, \ldots, g\}$ to
contain the $G_i(j)$-th cache line of $A$ for each $j \in \{1, 2,
\ldots, s\}$. That is, $G_i(j)$ denotes the index of the $j$-th
cache line from array $A$ contained in $U_i$.

Note that, to compute the index of the $j$-th cache line in $U_i$,
one needs only the value of $X[j]$. Thus the only metadata needed by
the algorithm to determine $U_1, \ldots, U_g$ is the array
$X$. If $|X| = s = \frac{n}{gb} \le \polylog(n)$, then
the algorithm is in place.
  
The algorithm performs an in-place (serial) partition on each
$U_i$ (and performs these partitions in parallel with one
another). In doing so, the algorithm, also collects
$v_{\text{min}}=\min_i{v_i}$, $v_{\text{max}}=\max_i{v_i}$, where
each $v_i$ with $i \in \{1, \ldots, g\}$ is defined to be the index
of the first successor in $A$ (or $n$ if no such successor
exists).\footnote{One can calculate $v_{\text{min}}$ and
  $v_{\text{max}}$ without explicitly storing each of $v_1, \ldots,
  v_{g}$ as follows. Rather than using a standard $g$-way parallel
  for-loop to partition each of $U_1, \ldots, U_{g}$, one can
  manually implement the parallel for-loop using a recursive
  divide-and-conquer approach. Each recursive call in the
  divide-and-conquer can then simply collect the maximum and minimum
  $v_i$ for the $U_i$'s that are partitioned within that recursive
  call. This adds $O(\log n)$ to the total span of the Partial
  Partition Step, which does not affect the overall span
  asymptotically. 
}

The array $A$ is now ``partially partitioned", i.e. $A[i]$ is a
predecessor for all $i \le v_{\text{min}}$, and $A[i]$ is a successor
for all $i > v_{\text{max}}$.

The second step of the Smoothed Striding Algorithm is to complete the
partitioning of $A[v_{\text{min}} + 1], \ldots, A[v_{\text{max}}]$. This can be done
in one of two ways: The \defn{Recursive Smoothed Striding Algorithm}
partitions $A[v_{\text{min}} + 1], \ldots, A[v_{\text{max}}]$ recursively using the
same algorithm (and resorts to a serial base case when the subproblem
is small enough that $g \le O(1)$); the \defn{Hybrid Smoothed Striding
  Algorithm} partitions $A[v_{\text{min}} + 1], \ldots, A[v_{\text{max}}]$ using the
in-place algorithm given in Theorem \ref{thminplace} with span $O(\log
n \log \log n)$. In general, the Hybrid algorithm yields better
theoretical guarantees on span than the recursive version; on the
other hand, the recursive version has the advantage that it is
simple to implement as fully in-place, and still achieves
polylogarithmic span. We analyze both algorithms in this section.

Detailed pseudocode for the Recursive Smoothed Striding Algorithm can
be found in Figure \ref{alg:parallelPartition_smoothedStriding}.

\begin{figure*}
  \caption{Smoothed Striding Algorithm}
  \label{alg:parallelPartition_smoothedStriding}
  \begin{algorithmic}
    \State \textbf{Recall:} 
    \State $A$ is the array to be partitioned, of length $n$. 
    \State We break $A$ into chunks, each consisting of $g$ cache lines of size $b$.
    \State We create $g$ groups $U_1,\ldots, U_g$ that each contain a single cache line from each chunk,
    \State $U_i$'s $j$-th cache line is the $(X[j]+i \bmod g + 1)$-th cache line in the $j$-th chunk of $A$.
    \State

    \Procedure{Get Block Start Index}{$X$, $g$, $b$, $i$, $j$}
      \Comment This procedure returns the index in $A$ of the start of $U_i's$ $j$-th block.
      \State\Return $b\cdot ((X[j] + i \bmod g) +(j-1)\cdot g)+1$
    \EndProcedure
    \State

    \Procedure{ParallelPartition}{$A$, $n$, $g$, $b$}
      \If{$g<2$}
        \State serial partition $A$
      \Else
        \For{$j \in \{1,2,\ldots,n/(gb)\}$}
          \State $X[j] \gets$ a random integer from $[1,g]$ 
        \EndFor

        \ForAll{$ i \in \{1,2,\ldots,g\}$ in parallel} \Comment We perform a serial partition on all $U_i$'s in parallel

          \State low $\gets$ GetBlockStartIndex($X$,$g$,$b$,$i$,$1$)
          \Comment low $\gets$ index of the first element in $U_i$
          \State high $\gets$ GetBlockStartIndex($X$,$g$,$b$,$i$,$n/(gb)$) + $b-1$
          \Comment high $\gets$ index of the last element in $U_i$

          \While{low $<$ high} 
            \While{A[low] $\leq$ pivotValue}
              \State low $\gets$ low$+1$
              \If{low $\bmod b \equiv 0$ } 
                \Comment Perform a block increment once low reaches the end of a block
                \State $k \gets $ number of block increments so far (including this one)
                \State low $\gets$ GetBlockStartIndex($X$,$g$,$b$,$i$,$k$)
                \Comment Increase low to start of block $k$ of $G_i$
              \EndIf
            \EndWhile
            \While{A[high] $>$ pivotValue}
              \State high $\gets$ high$-1$
              \If{high $\bmod b \equiv 1$} 
                \Comment Perform a block decrement once high reaches the start of a block
                \State $k \gets $ number of block decrements so far (including this one)
                \State $k' \gets n/(gb) - k$
                \State high $\gets$ GetBlockStartIndex($X$, $g$,$b$,$i$,$k'$) $+b-1$
                \Comment Decrease high to end of block $k'$ of $G_i$
              \EndIf
            \EndWhile
            \State Swap $A[\text{low}]$ and $A[\text{high}]$
          \EndWhile
        \EndFor
        \State Recurse on $A[v_{min}],\ldots,A[v_{max}-1]$
      \EndIf
    \EndProcedure
  \end{algorithmic}	
\end{figure*}

\paragraph{Algorithm Analysis} Our first proposition analyzes the Partial Partition Step.
\begin{proposition}
  \label{prop:generalResult}
  
  Let $\epsilon \in (0, 1/2)$ and $\delta \in (0, 1/2)$ such that
  $\epsilon \ge 1/\poly(n)$ and $\delta \ge 1/\polylog(n)$. 
  Suppose $s > \frac{\ln (n/\epsilon)}{\delta^2}$, and that each processor has
  a cache of size at least $s + c$ for a sufficiently large constant
  $c$.

  Then the Partial-Partition Algorithm achieves work $O(n)$; achieves
  span $O(b \cdot s)$; incurs $\frac{s+n}{b} + O(1)$ cache
  misses; and guarantees that with probability at least $1 -
  \epsilon$,
  $$v_{\text{max}}-v_{\text{min}} < 4 n \delta.$$
\end{proposition}

\begin{proof}
Since $\sum_i |U_i| = n$, and since the serial partitioning of each $U_i$
takes time $O(|U_i|)$, the total work performed by the algorithm is
$O(n)$.

To analyze cache misses, we assume without loss of generality
that array $X$ is pinned in each processor's cache (note, in
particular, that $|X| = s \le \polylog(n)$, and so $X$ fits in
cache). Thus we can ignore the cost of accesses to $X$. Note that
each $U_i$ consists of $s = \polylog n$ cache lines, meaning that
each $U_i$ fits entirely in cache. Thus the number of cache
misses needed for a thread to partition a given $U_i$ is just
$s$. Since there are $g$ of the $U_i's$, the total number of
cache misses incurred in partitioning all of the $U_i$'s is $g s
= n/b$. Besides these, there are $s/b$ cache misses for
instantiating the array $X$; and $O(1)$ cache misses for other
instantiating costs. This sums to $$\frac{n+s}{b}+O(1).$$

The span of the algorithm is $O(n/g + s) = O(b\cdot s)$, since the
each $U_i$ is of size $O(n / g)$, and because the initialization of
array $X$ can be performed in time $O(|X|) = O(s)$.

It remains to show that with probability $1-\epsilon$, $v_{\text{max}}
- v_{\text{min}} < 4n\delta$. Let $\mu$ denote the fraction of
elements in $A$ that are predecessors. For $i \in \{1, 2, \ldots,
g\}$, let $\mu_i$ denote the fraction of elements in $U_i$ that are
predecessors. Note that each $\mu_i$ is the average of $s$ independent
random variables $Y_i(1), \ldots, Y_i(s) \in [0, 1]$, where $Y_i(j)$
is the fraction of elements in the $G_i(j)$-th cache line of $A$ that
are predecessors. By construction, $G_i(j)$ has the same probability
distribution for all $i$, since $(X[j] + i) \pmod g$ is uniformly
random in $\mathbb{Z}_g$ for all $i$. It follows that $Y_i(j)$ has the
same distribution for all $i$, and thus that $\E[\mu_i]$ is
independent of $i$. Since the average of the $\mu_i$s is $\mu$, it
follows that $\E[\mu_i] = \mu$ for all $i \in \{1, 2, \ldots, g\}$.

Since each $\mu_i$ is the average of $s$ independent $[0, 1]$-random
variables, we can apply Hoeffding's inequality (i.e. a Chernoff Bound
for a random variable on $[0,1]$ rather than on $\{0,1\}$) to each
$\mu_i$ to show that it is tightly concentrated around its expected
value $\mu$, i.e.,
$$\Pr[|\mu_i - \mu| \geq \delta] < 2\exp(-2s\delta^2). $$

Since $s > \frac{\ln (n/\epsilon)}{\delta^2} \ge \frac{\ln (2n / (b\epsilon))}{2\delta^2}$, we find that for all $i \in
\{1,\ldots, g\}$,
$$\Pr[|\mu_i - \mu| \geq \delta] < 2\exp\Big({-2} \frac{\ln
  (2n/(b\epsilon))}{2\delta^2} \delta^2\Big) = \frac{\epsilon}{n/b} <
\frac{\epsilon}{g}. $$ By the union bound, it follows that with
probability at least $1 - \epsilon$, all of $\mu_1, \ldots, \mu_{g}$ are within $\delta$ of $\mu$.


To complete the proof we will show that the occurrence of the event
that all $\mu_y$ simultaneously satisfy $|\mu - \mu_y| < \delta$ implies
that $v_{\text{max}} - v_{\text{min}} \le 4n\delta$.



  Recall that $G_i(j)$ denotes the index within $A$ of the $j$ th
  cache-line contained in $U_i$. By the definition of $G_i(j)$,
  $$(j - 1)g + 1 \le G_i(j) \le jg.$$ Note that $A[v_i]$ will
  occur in the $\lceil s\mu_i \rceil$-th cache-line of $U_i$
  because $U_i$ is composed of $s$ cache lines. Hence $$(\lceil
  s\mu_i \rceil - 1) g b + 1 \le v_i \le \lceil s\mu_i \rceil g
  b,$$ which means that $$s\mu_i g b - gb + 1 \le v_i \le s\mu_i
  g b + gb.$$ Since $sgb = n$, it follows that $|v_i - n \mu_i|
  \le gb$. Therefore, $$|v_i - n \mu| < gb + n\delta.$$ This
  implies that the maximum of $|v_i - v_j|$ for any $i$ and $j$
  is at most, $2bg + 2\delta n$. Thus,
\begin{align*}
  v_{\text{max}} - v_{\text{min}} & \le 2n \left( \delta + \frac{bg}{n} \right)  = 2n \left( \delta + 1/s \right) \\
  & \le 2n \left(\delta + \frac{\delta^2}{\ln (n / \epsilon)}\right) < 4n\cdot\delta.
\end{align*}
\end{proof}

We use Proposition \ref{prop:generalResult} as a tool to analyze
the Recursive and the Hybrid Smoothed Striding Algorithms.

Rather than parameterizing the Partial Partition step in each
algorithm by $s$, Proposition \ref{prop:generalResult} suggests
that it is more natural to parameterize by $\epsilon$ and
$\delta$, which then determine $s$.

We will assume that both the hybrid and the recursive algorithms
use $\epsilon = 1/n^c$ for $c$ of our choice (i.e. with high
probability in $n$). Moreover, the Recursive Smoothed Striding
Algorithm continues to use the same value of $\epsilon$ within
recursive subproblems (i.e., the $\epsilon$ is chosen based on
the size of the first subproblem in the recursion), so that the
entire algorithm succeeds with high probability in $n$.

For both algorithms, the choice of $\delta$ results in a tradeoff
between cache misses and span. For the Recursive algorithm, we
allow for $\delta$ to be chosen arbitrarily at the top level of
recursion, and then fix $\delta  = \Theta(1)$ to be a
sufficiently small constant at all levels of recursion after the
first; this guarantees that we at least halve the size of the
problem between recursive iterations\footnote{In general, setting
$\delta = 1/8$ will result in the problem size being halved.
However, this relies on the assumption that $gb \mid n$, which is
only without loss of generality by allowing for the size of
subproblems to be sometimes artificially increased by a small
amount (i.e., a factor of $1 + gb / n = 1 + 1/s$). One can handle
this issue by decreasing $\delta$ to, say, $1/16$.}. Optimizing
$\delta$ further (after the first level of recursion) would only
affect the number of undesired cache misses by a constant factor.


Next we analyze the Hybrid Smoothed Striding Algorithm.
\begin{theorem}
	\label{thm:fullPartition}
  The Hybrid Smoothed Striding Algorithm using parameter
  $\delta\in(0,1/2)$ satisfying $\delta \ge 1/\polylog(n)$: has
  work $O(n)$; achieves span $$O\paren*{\log n \log\log n
  +\frac{b\log n}{\delta^2}},$$ with high probability in $n$; and
  incurs fewer than $$(n+O(n\delta))/b$$
cache misses with high probability in $n$.
\end{theorem}


An interesting corollary of Theorem \ref{thm:fullPartition}
concerns what happens when $b$ is small (e.g., constant) and we
choose $\delta$ to optimize span. 

\begin{corollary}[Corollary of Theorem \ref{thm:fullPartition}]
	\label{cor:fullPartition}
Suppose $b \le o(\log \log n)$. Then the Hybrid Smoothed Striding using $\delta
= \Theta\big(\sqrt{b/\log\log n}\big)$, achieves work $O(n)$, and with high
probability in $n$, achieves span $O(\log n \log\log n)$ and incurs fewer than
$(n+o(n))/b$ cache misses.
\end{corollary}

\begin{proof}[Proof of Theorem \ref{thm:fullPartition}]
  
  We analyze the Partial Partition Step using Proposition
  \ref{prop:generalResult}. Note that by our choice of $\epsilon$,
  $s=O\left(\frac{\log n}{\delta^2}\right)$.  The Partial Partition
  Step therefore has work $O(n)$, span $O\paren*{\frac{b\log
      n}{\delta^2}},$ and incurs fewer than
	$$\frac{n}{b}+O\paren*{\frac{\log n}{b\delta^2}}+O(1)$$ 
  cache misses.

  By Theorem \ref{thminplace}, the subproblem of partitioning of
  $A[v_{\text{min}} + 1], \ldots, A[v_{\text{max}}]$ takes work
  $O(n)$. With high probability in $n$, the subproblem has size
  less than $4n\delta$, which means that the subproblem achieves
  span $$O(\log n\delta \log\log n\delta) = O(\log n \log\log
  n),$$ and incurs at most $O(n \delta / b)$ cache misses.

  The total number of cache misses is therefore,
  	$$\frac{n}{b}+O\paren*{\frac{\log n}{b\delta^2} +
    \frac{n\delta}{b}}+O(1),$$ which since $\delta \ge 1 /
  \polylog(n)$, is at most $(n+O(n\delta))/b + O(1) \le (n + O(n
  \delta)) / b,$ as desired.
\end{proof}

\begin{proof}[Proof of Corollary \ref{cor:fullPartition}] We use
  $\delta = \sqrt{b/\log\log n}$ in the result proved in Theorem
  \ref{thm:fullPartition}. \\
  First note that the assumptions of Theorem
  \ref{thm:fullPartition} are satisfied because
  $O(\sqrt{b/\log\log n}) > 1 / \polylog(n).$
	The algorithm achieves work $O(n)$. 
	With high probability in $n$ the algorithm achieves span 
	$$O\paren*{\log n \log\log n +\frac{b\log n}{\delta^2}} = O(\log n\log\log n).$$
	With high probability in $n$ the algorithm incurs fewer than 
	$$(n+O(n\delta))/b = (n+O(n\sqrt{b/\log\log n}))/b$$ 
	cache misses.
	By assumption $\sqrt{b/\log\log n} = o(1)$, so this reduces to 
	$(n+o(n))/b$
	cache misses, as desired.
\end{proof}


The next theorem analyzes the span of the Recursive Smoothed Striding Algorithm.
\begin{theorem}
	\label{thm:groupedPartitionAlg}
	With high probability in $n$, the Recursive Smoothed Striding
        algorithm using parameter $\delta \in(0,1/2)$ satisfying
        $\delta \ge 1 / \polylog(n)$: achieves work $O(n)$, attains span
	$$O\left(b\left(\log^2 n + \frac{\log n}{\delta^2}\right)\right),$$
	and incurs $(n+O(n \delta))/b$ cache misses. 
\end{theorem}

A particularly natural parameter setting for the Recursive
algorithm occurs at $\delta = 1 / \sqrt{\log n}$.
\begin{corollary}[Corollary of Theorem \ref{thm:groupedPartitionAlg}]
  \label{cor:groupedPartitionAlg}
  With high probability in $n$, the
  Recursive Smoothed Striding Algorithm using parameter
  $\delta=1/\sqrt{\log n}$: achieves work $O(n)$, attains span
  $O(b\log^2 n)$, and incurs $(1 + o(1))n/b$ cache misses. 
\end{corollary}

\begin{proof}[Proof of Theorem \ref{thm:groupedPartitionAlg}]
  To avoid confusion, we use $\delta'$, rather than $\delta$, to
  denote the constant value of $\delta$ used at levels of recursion
  after the first.

  By Proposition \ref{prop:generalResult}, the top level of the algorithm
  has work $O(n)$, span $O\Big(b\frac{\log n}{\delta^2}\Big),$ and
  incurs $\frac{s+n}{b} + O(1)$ cache misses.  The recursion reduces
  the problem size by at least a factor of $4\delta$, with high
  probability in $n$.

  At lower layers of recursion, with high probability in $n$, the
  algorithm reduces the problem size by a factor of at least
  $1/2$ (since $\delta$ is set to be a sufficiently small
  constant). For each $i > 1$, it follows that the size of the
  problem at the $i$-th level of recursion is at most $O(n \delta
  / 2^i)$.
  
  The sum of the sizes of the problems after the first level of
  recursion is therefore bounded above by a geometric series summing to at most $O(n
  \delta)$. This means that the total work of the algorithm is at most
  $O(n\delta) + O(n) \le O(n)$.

  Recall that each level $i > 1$ uses $s =
  \frac{\ln(2^{-i}n\delta'/b)}{\delta'^2}$, where $\delta' =
  \Theta(1)$. It follows that level $i$ uses $s \le O(\log n)$.
  Thus, by Proposition \ref{prop:generalResult}, level $i$
  contributes $O(b\cdot s)=O(b \log n)$ to the span.  Since there
  are at most  $O(\log n)$ levels of recursion, the total span in
  the lower levels of recursion is at most $O(b\log^2 n)$, and
  the total span for the algorithm is at most,
  $$O\left(b\left(\log^2 n + \frac{\log
  n}{\delta^2}\right)\right).$$
        
  To compute the total number of cache misses of the algorithm,
  we add together $(n+s)/b+O(1)$ for the top level, and then, by
  Proposition \ref{prop:generalResult}, at most $$\sum_{0 \leq i<
  O(\log n)}\frac{1}{b} O\paren*{2^{2-i}n\delta + \log n} \le
  O\left(\frac{1}{b} (n \delta + \log^2 n)\right).$$ for lower
  levels. Thus the total number of cache misses for the algorithm
  is, $$\frac{1}{b}\left(n+\frac{\log n}{\delta^2 }\right) +
  O(n\delta + \log^2 n) / b = (n+O(n\delta))/b.$$ 
\end{proof}

\begin{proof}[Proof of Corollary \ref{cor:groupedPartitionAlg}]
  By Theorem \ref{thm:groupedPartitionAlg}, with high probability
  in $n$, the algorithm has work $O(n)$, the algorithm has span
  $$O\left(b\left(\log^2 n + \frac{\log
  n}{\delta^2}\right)\right) = O(b\log^2 n),$$ and the algorithm
  incurs $$(n+O(n\delta))/b = (n+O(n/\sqrt{\log n}))/b =
  (n+o(n))/b$$ cache misses.
\end{proof}


\section{Performance Comparisons}\label{secexp}

In this section, we implement the techniques from Section
\ref{secalg} and Section \ref{sec:smoothing} to build space-efficient and
in-place parallel-partition functions.

Each implementation considers an array of $n$ 64-bit integers, and
partitions them based on a pivot. The integers in the array are
initially generated so that each is randomly either larger or smaller
than the pivot.

In Subsection \ref{subsecclassic}, we evaluate the techniques in
Section \ref{secalg} for transforming the standard
parallel-prefix-based partition algorithm into an in-place
algorithm. We compare the performance of three parallel-partition
implementations: (1) The \defn{high-space} implementation which
follows the standard parallel-partition algorithm exactly; (2) a
\defn{medium-space} implementation which reduces the space used for
the Parallel-Prefix phase; and (3) a \defn{low-space} implementation
which further eliminates the auxiliary space used in the Reordering
phase of the algorithm. The low-space implementation still uses a
small amount of auxiliary memory for the parallel-prefix, storing
every $O(\log n)$-th element of the parallel-prefix array explicitly
rather than using the implicit-storage approach in Section
\ref{secalg}. Nonetheless the space consumption is several orders of
magnitude smaller than the original algorithm.

In addition to achieving a space-reduction, the better cache-behavior
of the low-space implementation allows for it to achieve a speed
advantage over its peers, in some cases completing roughly twice as
fast as the medium-space implementation and four times as fast as the
low-space implementation. We show that all three implementations are
bottlenecked by memory throughput, however, suggesting that the cache-optimal Smoothed Striding Algorithm can do better.

In Subsection \ref{subsecstrided}, we evaluate the performance of the
Recursive Smoothed Striding Algorithm and the Strided
Algorithm. Unlike the algorithms described above, the implementations
of both of these algorithms are fully in-place, meaning that the total
space overhead is only $\polylog n$. The cache efficiency of these two
algorithms allows for them to achieve substantially better scaling
than their parallel-prefix-based counterparts. The Strided Algorithm
tends to slightly outperform the Smoothed Striding Algorithm, though
on 18 threads their performance is within 15\% of one-another. We
conclude that the Smoothed Striding Algorithm allows for one to obtain
empirical performance comparable to that of the Strided Algorithm,
while simultaneously achieving the provable guarantees on span and
cache-efficiency missing from the original Strided Algorithm.



\paragraph{Machine Details}
Our experiments are performed on a two-socket machine with eighteen
2.9 GHz Intel Xeon E5-2666 v3 processors. To maximize the memory
bandwidth of the machine, we use a NUMA memory-placement policy in
which memory allocation is spread out evenly across the nodes of the
machine; this is achieved using the \emph{interleave=all} option in
the Linux \emph{numactl} tool \cite{Kleen05}. Worker threads in our
experiments are each given their own core, with hyperthreading
disabled.

Our algorithms are implemented using the CilkPlus task parallelism
library in C++. The implementations avoid the use of concurrency
mechanisms and atomic operations, but do allow for concurrent reads to
be performed on shared values such as $n$ and the pointer to the input
array. Our code is compiled using g++ 7.3.0, with \emph{march=native}
and at optimization level three. 

Our implementations are available on GitHub.

\subsection{Comparing Parallel-Prefix-Based Algorithms}\label{subsecclassic}

In this section, we compare four partition implementations,
incorporating the techniques from Section \ref{secalg} in order to
achieve space efficiency:
\begin{itemize}[leftmargin = .15in]
\item \emph{A Serial Baseline:} This uses the serial in-place
  partition implementation from GNU Libc quicksort, with minor
  adaptations to optimize it for the case of sorting 64-bit integers
  (i.e., inlining the comparison function, etc.).
\item \emph{The High-Space Parallel Implementation:} This uses the
  standard parallel partition algorithm \cite{Blelloch96,AcarBl16}, as
  described in Section \ref{secprelim}. The space overhead is roughly
  $2n$ eight-byte words.
\item \emph{The Medium-Space Parallel Implementation:} Starting with
  the high-space implementation, we reduce the space used by the
  Parallel-Prefix phase by only constructing every $O(\log n)$-th
  element of the prefix-sum array $B$, as in Section
  \ref{secalg}. (Here $O(\log n)$ is hard-coded as 64.) The array $B$
  is initialized to be of size $n / 64$, with each component equal to
  $\sum_{i=1}^{64} \dec(A[64 (i-1)+1])$, and then a parallel prefix sum is computed on
  the array $B$. Rather than implicitly encoding the elements of $B$ in
  $A$, we use an auxiliary array of size $n / 64$ to explicitly store
  the prefix sums.

  The algorithm
  has a space overhead of $\frac{n}{32} + n$ eight-byte
  words.\footnote{In addition to the auxiliary array of size $n / 64$,
    we use a series of smaller arrays of sizes $n / 128, n / 256,
    \ldots$ in the recursive computation of the prefix sum. The
    alternative of performing the parallel-prefix sum in place, as in
    Section \ref{secalg}, tends to be less cache-friendly in
    practice.}
\item \emph{The Low-Space Parallel Implementation:}
Starting with the medium-space implementation, we make the reordering
phase completely in-place using the preprocessing technique in Section
\ref{secalg}.\footnote{Depending on whether the majority of elements
  are predecessors or successors, the algorithm goes down separate
  (but symmetric) code paths. In our timed experiments we test only
  with inputs containing more predecessors than successors, since this
  the slower of the two cases (by a very slight amount) for our
  implementation.} The only space overhead in this implementation is
the $\frac{n}{32}$ additional 8-byte words used in the prefix sum.
\end{itemize}


We remark that the ample parallelism of the low-space algorithm makes
it so that for large inputs the value $64$ can easily be increased
substantially without negatively effecting algorithm performance. For
example, on an input of size $2^{28}$, increasing it to $4096$ has
essentially no effect on the empirical runtime while bringing the
auxiliary space-consumption down to a $\frac{1}{2048}$-fraction of the
input size. (In fact, the increase from 64 to 4096 results in roughly
a 5\% speedup.)

\paragraph{An Additional Optimization for The High-Space Implementation}
The optimization of reducing the prefix-sum by a factor of $O(\log n)$
at the top level of recursion, rather than simply by a factor of two,
can also be applied to the standard parallel-prefix algorithm when
constructing a prefix-sum array of size $n$. Even without the space
reduction, this reduces the (constant) overhead in the parallel prefix
sum, while keeping the overall span of the parallel-prefix operation
at $O(\log n)$. We perform this optimization in the high-space
implementation.

\paragraph{Performance Comparison}
\begin{figure*}
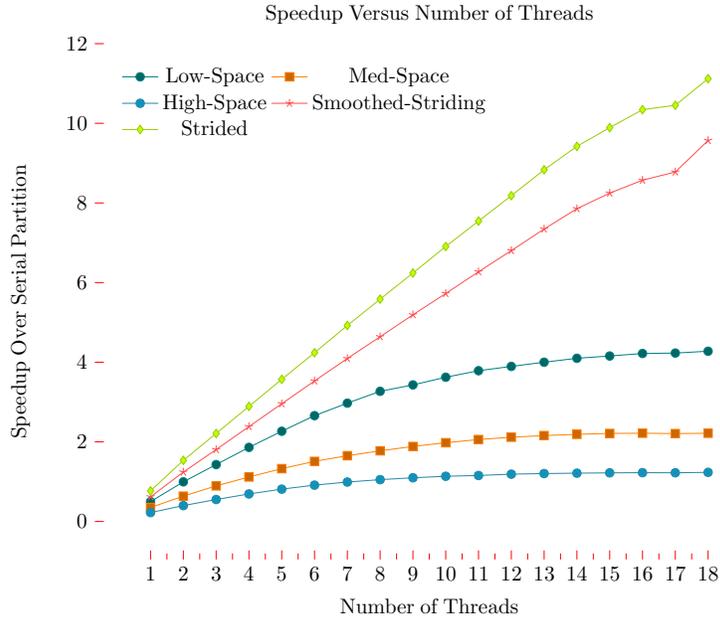

  \begin{center}
    \CILKtable 
  \end{center}
    \caption{For a fixed table-size $n = 2^{30}$, we compare each
      implementation's runtime to the Libc serial baseline, which takes 3.9
      seconds to complete (averaged over five trials). The $x$-axis
      plots the number of worker threads being used, and the $y$-axis
      plots the multiplicative speedup over the serial baseline. Each
      time (including the serial baseline) is averaged over five trials.}
      \label{tablecilk}
\end{figure*}

\begin{figure*}
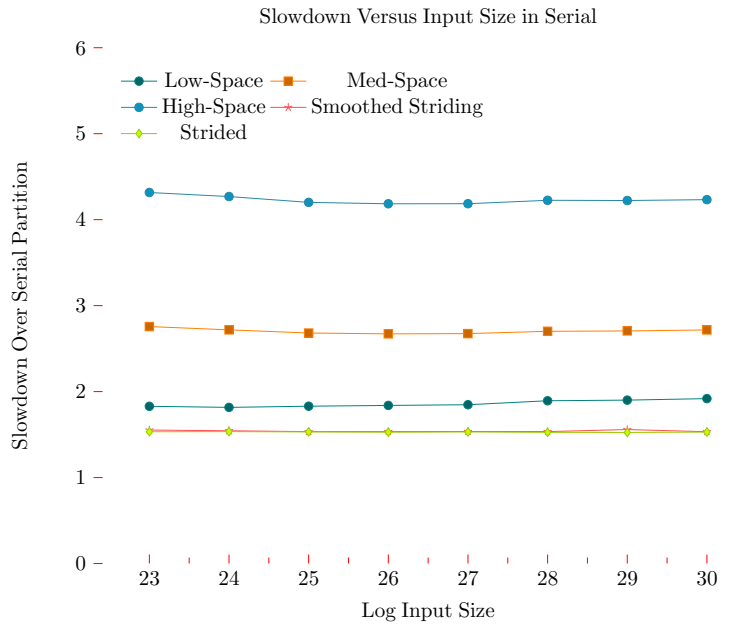

  \begin{center}
    \serialtable
  \end{center}
  \caption{We compare the performance of the implementations in
    serial, with no scheduling overhead. The $x$-axis is the
    log-base-$2$ of the input size, and the $y$-axis is the
    multiplicative slowdown when compared to the Libc serial baseline. Each
    time (including the baseline) is averaged over five
    trials.}
  \label{tableserial}
\end{figure*}

Figure \ref{tablecilk} graphs the speedup of the each of the parallel
algorithms over the serial algorithm, using varying numbers of worker
threads on an 18-core machine with a fixed input size of $n =
2^{30}$. Both space optimizations result in performance improvements,
with the low-space implementation performing almost twice as well as
the medium-space implementation on eighteen threads, and almost four
times as well as the high-space implementation. 

Figure \ref{tableserial} compares the performances of the
implementations in serial. Parallel-for-loops are replaced with serial
for-loops to eliminate scheduler overhead. As the input-size varies,
the ratios of the runtimes vary only slightly. The low-space
implementation performs within a factor of roughly 1.9 of the serial
implementation. As in Figure \ref{tablecilk},
both space optimizations result in performance improvements.

\paragraph{The Source of the Speedup} If we compare the number of
instructions performed by the three parallel implementations,
then the medium-space algorithm would seem to be the clear
winner. Using Cachegrind to profile the number of instructions
performed in a (serial) execution on an input of size
$2^{28}$,\footnote{This smaller problem size is used to
compensate for the fact that Cachegrind can be somewhat slow.}
the high-space, medium-space, and low-space implementations
perform 4.4 billion, 2.9 billion, and 4.6 billion instructions,
respectively.

Cache misses tell a different story, however. Using Cachegrind to
profile the number of top-level cache misses in a (serial) execution
on an input of size $2^{28}$, the high-space, medium-space, and
low-space implementations incur 305 million, 171 million, and 124
million cache misses, respectively.

To a first approximation, the number of cache misses by each algorithm
is proportional to the number of times that the algorithm scans
through a large array. By eliminating the use of large auxiliary
arrays, the low-space implementation has the opportunity to achieve a
reduction in the number of such scans. Additionally, the low-space
algorithm allows for steps from adjacent phases of the algorithm to
sometimes be performed in the same pass. For example, the enumeration
of the number of predecessors and the top level of the Preprocessing
phase can be performed together in a single pass on the input
array. Similarly, the later levels of the Preprocessing phase (which
focus on only one half of the input array) can be combined with the
construction of (one half of) the auxiliary array used in the Parallel
Prefix Sum phase, saving another half of a pass.


\paragraph{The Memory-Bandwidth Limitation}
\begin{figure*}
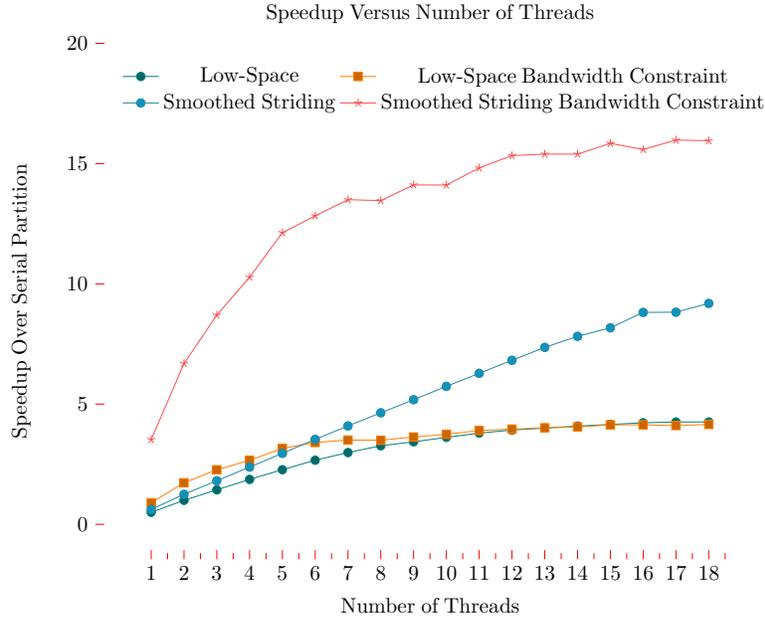

  \begin{center}
    \partitionbandwidthboundtable
  \end{center}
  \caption{We compare the performances of the low-space and
    Smoothed Striding parallel-partition algorithms to their ideal
    performance determined by memory-bandwidth constraints on inputs
    of size $2^{30}$. The $x$-axis is the number of worker threads,
    and the $y$-axis is the multiplicative speedup when compared to
    the Libc serial baseline (which is computed by an average over five
    trials). Each data-point is averaged over five trials.}
  \label{tablebandwidth}
\end{figure*}

The comparison of cache misses suggests that performance is
bottlenecked by memory bandwidth. To evaluate whether this is the
case, we measure for each $t \in \{1, \ldots, 18\}$ the memory
throughput of $t$ threads attempting to scan through disjoint
portions of a large array in parallel. We measure two types of
bandwidth, the \defn{read-bandwidth}, in which the threads are
simply trying to read from the array, and the \defn{read/write
bandwidth}, in which the threads are attempting to immediately
overwrite entries to the array after reading them. Given
read-bandwidth $r$ bytes/second and read/write bandwidth $w$
bytes/second, the time needed for the low-space algorithm to
perform its memory operations on an input of $m$ bytes will be
roughly $3.5 m / w + .5m / r$ seconds.\footnote{A naive
implementation of the algorithm would require roughly $m / r$
time to count the number of predecessors, followed by $2m / w$
time to perform the Preprocessing Phase, followed by roughly $m /
r$ time to perform the Parallel Prefix Sum Phase, and then
roughly $1.5 m/w$ time for the In-Place Reordering Phase. As
described in the previous paragraph, however, the counting of
predecessors and the Parallel Prefix Sum phase can both be
overlapped with the Preprocessing phase so that their total added
contribution to the Memory-Bandwidth Limitation is only $.5 m /
r$.}  We call this the \defn{bandwidth constraint}. No matter how
optimized the implementation of the low-space algorithm is, the
bandwidth constraint serves as a hard lower bound for the running
time.\footnote{Empirically, on an array of size $n = 2^{28}$, the
  total number of cache misses is within $8\%$ of what this
assumption would predict, suggesting that the bandwidth
constraint is within a small amount of the true bandwidth-limited
runtime.}

Figure \ref{tablebandwidth} compares the time taken by the low-space
algorithm to the bandwidth constraint as the number of threads $t$
varies from $1$ to $18$. As the number of threads grows, the algorithm
becomes bandwidth limited, achieving its best possible parallel
performance on the machine. The algorithm scales particularly well on
the first socket of the machine, achieving a speedup on nine cores of
roughly six times better than its performance on a single core, and
then scales more poorly on the second socket as it becomes
bottlenecked by memory bandwidth.

\paragraph{Implementation Details}
In each implementation, the parallelism is achieved through simple
parallel-for-loops, with one exception at the beginning of the
low-space implementation, when the number of predecessors in the input
array is computed. Although CilkPlus Reducers (or OpenMP Reductions)
could be used to perform this parallel summation within a
parallel-for-loop \cite{FrigoLe09}, we found a slightly more ad-hoc
approach to be faster: Using a simple recursive structure, we manually
implemented a parallel-for-loop with Cilk Spawns and Syncs, allowing
for the summation to be performed within the recursion.

\subsection{Comparing the Smoothed Striding and Strided Algorithms} \label{subsecstrided}

In this section we consider the performance of the Strided Algorithm
and the Recursive Smoothed Striding Algorithm. Past work
\cite{Frias08} found that, on large numbers of threads, the Strided
Algorithm has performance close to that of other non-EREW state-of-the
art partition algorithms (i.e., within 20\% of the best
atomic-operation based algorithms). The Strided Algorithm does not
offer provable guarantees on span and cache-efficiency, however; and
indeed, the reason that the algorithm cannot recurse on the subarray
$A[v_{\text{min}} + 1], \ldots, A[v_{\text{max}}]$ is that the subarray has been
implicitly constructed to be worst-case for the algorithm. In this
subsection, we show that, with only a small loss in performance, the
Smoothed Striding Algorithm can be used to achieve provable guarantees
on arbitrary inputs. We remark that we do not make any attempt to
generate worst-case inputs for the Strided Algorithm (in fact the
random inputs that we use are among the only inputs for which the
Strided Algorithm does exhibit provable guarantees!).

Figures \ref{tableserial} and \ref{tablecilk} evaluate the performance
of the Smoothed Striding and Strided algorithms in serial and in
parallel. On a single thread, the Smoothed Striding and Strided
algorithms perform approximately 1.5 times slower than the Libc-based
serial implementation baseline. When executed on multiple threads, the
performances of the Smoothed Striding and Strided Algorithms scale
close to linearly in the number of threads. On 18 threads, the
Smoothed Striding Algorithm achieves a $9.6 \times $ speedup over the
Libc-based Serial Baseline, and the Strided Algorithm achieves an
$11.1 \times$ speedup over the same baseline.

The nearly-ideal scaling of the two algorithms can be explained by
their cache behavior. Whereas the parallel-prefix-based algorithms
were bottlenecked by memory bandwidth, Figure \ref{tablebandwidth}
shows that the same is no longer true for the Smoothed Striding
Algorithm. The figure compares the performance of the Smoothed
Striding Algorithm to the minimum time needed simple to read and overwrite
each entry of the input array using $18$ concurrent threads without any other
computation (i.e., the memory bandwidth constraint). On 18 threads, the time required by the memory bandwidth
constraint constitutes $58\%$ of the algorithm's total running time.

\paragraph{NUMA Effects}
We remark that the use of the Linux \emph{numactl} tool \cite{Kleen05}
to spread memory allocation evenly across the nodes of the machine is
necessary to prevent the Smoothed Striding Algorithm and the Strided
Algorithm from being bandwidth limited. For example, if we replicate
the 18-thread column of Figure \ref{tablebandwidth} without using
\emph{numactl}, then the speedup of the Smoothed Striding Algorithm is
8.2, whereas the memory-bandwidth bound for maximum possible speedup
is only slightly larger at $10.2$.

\paragraph{Implementation Details} Both algorithms use $b = 512$. The Smoothed Striding Algorithm uses
slightly tuned $\epsilon, \delta$ parameters similar to those outlined
in Corollary \ref{cor:groupedPartitionAlg}. Although $v_{\text{min}}$
and $v_{\text{max}}$ could be computed using CilkPlus Reducers
\cite{FrigoLe09}, we found it advantageous to instead manually
implement the parallel-for-loop  in the Partial Partition step with Cilk Spawns and Syncs, and to
compute $v_{\text{min}}$ and $v_{\text{max}}$ within the recursion.

\paragraph{Example Application: A Full Quicksort}

\begin{figure*}
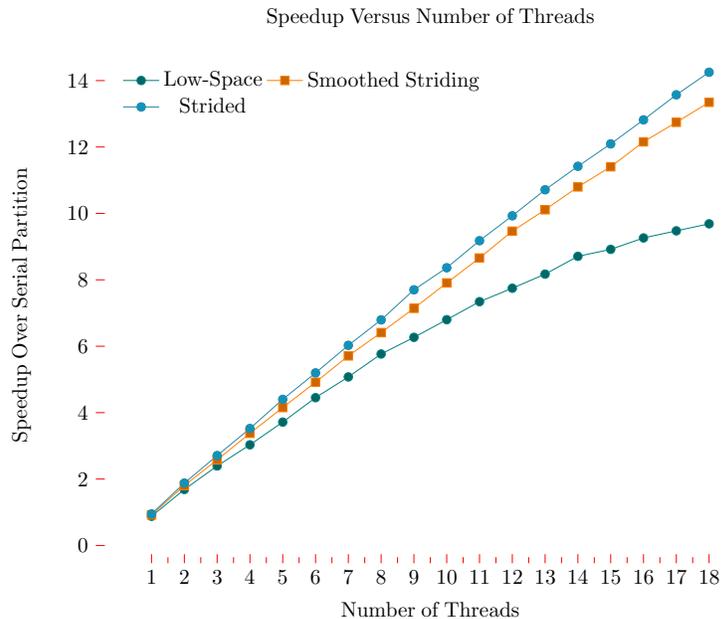

  \begin{center}
    \CILKsorttable
  \end{center}
  \caption{We compare the performance of the low-space and high-span
    sorting implementations running on varying numbers of threads and
    input sizes. The $x$-axis is the number of worker threads and the
    $y$-axis is the multiplicative speedup when compared to the Libc serial
    baseline. Each time (including each serial baseline)
    is averaged over five trials.}
  \label{tablesort}
\end{figure*}

In Figure \ref{tablesort}, we graph the performance of a parallel
quicksort implementation using the low-space parallel-prefix-based
algorithm, the Smoothed Striding Algorithm, and the Strided
Algorithm. We compare the algorithm performances 
with varying numbers of worker threads and input sizes to GNU Libc quicksort; the input
array is initially in a random permutation.

Our parallel quicksort uses the parallel-partition algorithm at the
top levels of recursion, and then swaps to the serial-partitioning
algorithm once the input size has been reduced by at least a factor of
$8p$, where $p$ is the number of worker threads. By using the
serial-partitioning algorithm on the small recursive subproblems we
avoid the overhead of the parallel algorithm, while still achieving
parallelism between subproblems. Small recursive problems also exhibit
better cache behavior than larger ones, reducing the effects of
memory-bandwidth limitations on the performance of the parallel
quicksort, and further improving the scaling.

\section{Conclusion and Open Questions}\label{sec:open}

Parallel partition is a fundamental primitive in parallel algorithms
\cite{Blelloch96,AcarBl16}. Achieving faster and more space-efficient
implementations, even by constant factors, is therefore of high
practical importance. Until now, the only space-efficient algorithms
for parallel partition have relied extensively on concurrency
mechanisms or atomic operations, or lacked provable performance
guarantees. If a parallel function is going to be invoked within a large
variety of applications, then provable guarantees are highly
desirable. Moreover, algorithms that avoid the use of concurrency
mechanisms tend to scale more reliably (and with less dependency on
the particulars of the underlying hardware).

In this paper, we have shown that, somewhat surprisingly, one can
adapt the classic parallel algorithm to completely eliminate the use
of auxiliary memory, while still using only exclusive read/write
shared variables, and maintaining a polylogarithmic span. Although the
superior cache performance of the low-space algorithm results in
practical speedups over its out-of-place counterpart, both algorithms
remain far from the state-of-the art due to memory bandwidth bottlenecks. To close this gap, we also
presented a second in-place algorithm, the Smoothed Striding
Algorithm, which achieves polylogarithmic span while guaranteeing
provably optimal cache performance up to low-order factors. The
Smoothed Striding Algorithm introduces randomization techniques to the
previous (blocked) Striding Algorithm of \cite{Frias08, FrancisPa92},
which was known to perform well in practice but which previously
exhibited poor theoretical guarantees. Our implementation of the
Smoothed Striding Algorithm is fully in-place, exhibits
polylogarithmic span, and has optimal cache performance.

Our work prompts several theoretical questions. Can fast
space-efficient algorithms with polylogarithmic span be found for
other classic problems such as randomly permuting an array
\cite{Anderson90, AlonsoSc96, ShunGu15}, and integer sorting
\cite{Rajasekaran92, HanHe12, AlbersHa97, Han01, GerbessiotisSi04}?
Such algorithms are of both theoretical and practical interest, and
might be able to utilize some of the techniques introduced in this
paper.

Another important direction of work is the design of in-place parallel
algorithms for sample-sort, the variant of quicksort in which multiple
pivots are used simultaneously in each partition. Sample-sort can be
implemented to exhibit fewer cache misses than quicksort, which is
especially important when the computation is memory-bandwidth
bound. The known in-place parallel algorithms for sample-sort rely
heavily on atomic instructions \cite{AxtmannWi17} (even requiring
128-bit compare-and-swap instructions). Finding fast algorithms that
use only exclusive-read-write memory (or
concurrent-read-exclusive-write memory) is an important direction of
future work.

\bibliographystyle{plain}
\bibliography{paper}

\end{document}